\documentclass[10pt,journal]{IEEEtran} % Aptara syntax

\usepackage[display]{algorithm}
\usepackage[noend]{algorithmic}
\usepackage{enumitem}
\usepackage{mdwlist}
\usepackage{framed}
\usepackage{xspace}
\usepackage[dvips]{graphicx}
\usepackage{subfigure}
\usepackage{amsthm}
\usepackage[cmex10]{amsmath}
\usepackage{amssymb, amsfonts} 
\usepackage{array}
\usepackage{cases}
\usepackage{stfloats}
\usepackage{float}
\usepackage{widetext}
\usepackage{verbatim}

\newtheorem{definition}{Definition}

\makeatletter
\renewcommand{\section}{\@startsection {section}{1}{\z@}%
  {-1.5ex \@plus -0.5ex \@minus -.1ex}%
  {0.5ex \@plus.1ex}%
  {\normalfont\Large\bfseries}}
\renewcommand{\subsection}{\@startsection{subsection}{2}{\z@}%
  {-1.25ex\@plus -0.5ex \@minus -.1ex}%
  {0.5ex \@plus .1ex}%
  {\normalfont\large\bfseries}}
\renewcommand{\subsubsection}{\@startsection{subsubsection}{3}{\z@}%
  {-1.25ex\@plus -0.5ex \@minus -.1ex}%
  {0.5ex \@plus .1ex}%
  {\normalfont\normalsize\bfseries}}
\renewcommand{\paragraph}{\@startsection{paragraph}{4}{\z@}%
  {1.25ex \@plus.5ex \@minus.1ex}%
  {-0.5em}%
  {\normalfont\normalsize\bfseries}}
\makeatother

\floatname{algorithm}{Procedure Set}

\newtheorem{thm}{Theorem}

\def\etal{et\ al.\xspace}

\begin{document}

% Page heads
\author{Sumeet Bajaj, Anrin Chakraborti, Radu Sion
\thanks {Sumeet Bajaj, Anrin Chakraborti and Radu Sion are with the Department
of Computer Science, StonyBrook University, Stony Brook,
NY, 11790.}
}
        
% Title portion
\title{Practical Foundations of History Independence}
%\title{Practical History Independence: From Theory to Implementation}
\maketitle

% NOTE! Affiliations placed here should be for the institution where the
%       BULK of the research was done. If the author has gone to a new
%       institution, before publication, the (above) affiliation should NOT be changed.
%       The authors 'current' address may be given in the "Author's addresses:" block (below).
%       So for example, Mr. Abdelzaher, the bulk of the research was done at UIUC, and he is
%       currently affiliated with NASA.

%\IEEEtitleabstractindextext{%
\begin{abstract}
The way data structures organize data is often a function of the sequence of past operations. The organization of data is referred to as the data structure's state, and the sequence of past operations constitutes the data structure's history. A data structure state can therefore be used as an oracle to derive information about its history. As a result, for history-sensitive applications, such as privacy in e-voting, incremental signature schemes, and regulatory compliant data retention; it is imperative to conceal historical information contained within data structure states.

Data structure history can be hidden by making data structures history independent. In this paper, we explore how to achieve history independence.

We observe that current history independence notions are significantly limited in number and scope. There are two existing notions of history independence -- weak history independence (WHI) and strong history independence (SHI). WHI does not protect against insider adversaries and SHI mandates canonical representations, resulting in inefficiency.

We postulate the need for a broad, encompassing notion of history independence, which can capture WHI, SHI, and a broad spectrum of new history independence notions. To this end, we introduce $\Delta$history independence ($\Delta$HI), a generic game-based framework that is malleable enough to accommodate existing and new history independence notions. 

As an essential step towards formalizing $\Delta$HI, we explore the concepts of abstract data types, data structures, machine models, memory representations and history independence. Finally, to bridge the gap between theory and practice, we outline a general recipe for building end-to-end, history independent systems and demonstrate the use of the recipe in designing two history independent file systems.

\end{abstract}


%\category{C.2.2}{Computer-Communication Networks}{Network Protocols}
%\category{H.3.2}{Information Storage and Retrieval}{Information Storage}[File organization]

%\terms{Security, Design, Theory}

%\acmformat{Sumeet Bajaj, Radu Sion, 2014. Theory to Practice of History Independence.}
% At a minimum you need to supply the author names, year and a title.
% IMPORTANT:
% Full first names whenever they are known, surname last, followed by a period.
% In the case of two authors, 'and' is placed between them.
% In the case of three or more authors, the serial comma is used, that is, all author names
% except the last one but including the penultimate author's name are followed by a comma,
% and then 'and' is placed before the final author's name.
% If only first and middle initials are known, then each initial
% is followed by a period and they are separated by a space.
% The remaining information (journal title, volume, article number, date, etc.) is 'auto-generated'.

%Author's addresses: S. Bajaj {and} R. Sion, Computer Science Department,
%Stony Brook University.
%\end{bottomstuff}

%\maketitle

\section{Introduction}
\label{hitheory:intro:hi}

%Data is an essential component in any computation. Both functional and
%computational efficiency often depend on the techniques used to store and
%manage data. We commonly refer to the constructs used for storage and
%organization of data as ``data structures''.

Data structures are commonly used constructs to store and retrieve
data in systems.
However, data structures carry more information than the raw data they
organize. One aspect of this information is the history 
leading to the data structure's current state \cite{golovinthesis}.
%Leakage of such historical information can have serious implications for privacy.

%%Leakage of historical information from data structure states can violate regulatory compliance.
%For example, the current layout of data blocks on disk is a function of the sequence 
%%and timing
%of previous writes (i.e., history of past operations)
%to file system or database search indexes.
%Questions such as ``was John's record ever in the HIV patients'
%dataset'' can then be answered much more accurately than guessing by simply
%looking at the storage layout of the search index on disk, since the
%layout could be different depending on whether John has previously been
%in the data set or not. This potentially infringes on data retention
%laws that mandate secure, irrecoverable erasure \cite{cfr240,pipeda,eu-drd}.
%%For example, consider a voting application. If the underlying data structures
%%used to store individual votes reveal the order in which the votes were
%%cast (or modified), it may be considered as a violation of voter privacy.

Concealing historical information contained within data structure
states is necessary for incremental signature schemes
\cite{naorcuckoo} and for privacy in voting systems
\cite{naorcuckoo,blellochhashing,molnarprom,moranwriteonce}.
%and for un-traceable deletion \cite{ficklebase}.
%and regulatory-compliant secure deletion \cite{molnarprom,ficklebase}.
Therefore, the need arises for data structures that reveal no
information about the history that led to their current
state other than what is inherently visible from the data.
History independence \cite{hartlinecharacterizinghi}
has been devised to enable the design of such data structures and they are
termed as \emph{``history independent data structures''}.
%The history here is the sequence of past operations.

%The concept of \emph{``history independent data structures''} can be
%briefly explained as follows.
%The specification of data organization techniques is done via abstract
%data types (ADT). Roughly, an ADT is a set of states and a set of
%operations that take the ADT from the current state to a new state.
%The key characteristic of an ADT is that it specifies operations
%independently of any specific implementation.
%A data structure then, is a concrete implementation of
%an ADT in a specific machine model. A data structure is history independent
%if the machine state representing an ADT state
%reveals no additional
%information about the history of past operations than what is already 
%available from the corresponding ADT state.

%This description of history independence is intuitive, but imprecise.
%In order to design data structures (and complete systems) with history independent 
We have identified the role of history independence in designing systems
%Our goal is to design systems
that are compliant with data retention regulations \cite{cfr240}.
Retention regulations desire that once data is deleted, no evidence about
the past existence of deleted data must be recoverable.
Such a deletion cannot be achieved
by simply overwriting data as in secure deletion \cite{diesburgsurvey}.
%or by deploying encryption with ephemeral keys \cite{hifs}.
%This is because the previous existence of deleted records impacts
%the current system state implicitly at all layers.
This is because overwriting does not eliminate the effects that
previous existence of delete data leaves on the current system
state.
Even after secure deletion, the current state can be used as an
oracle to derive information about the past existence of deleted
records.
For example, the current organization of data blocks on disk is a function of the sequence 
of previous writes to file system or to database search indexes.
The organization could be different depending on whether a particular record
was deleted in the past or was never inserted in the data set.
Therefore, questions about history, such as ``was John's record ever in the HIV patients'
dataset'' can be answered much more accurately than guessing by simply
looking at the search index organization on disk since the
organization could be different depending on whether John has previously been
in the data set or not.
The inference of past existence of deleted data violates data retention regulations. \\
However, in order to architect systems with history independent characteristics and
to prove history independence, we need a formal notion of data structures,
of data structure states, and of history independence itself.
In this paper we first formalize all necessary concepts
and understand history independence from a theoretical perspective
(Sections \ref{hitheory:preliminaries} - \ref{hitheory:hi:generalizinghi}).
Then, in Section \ref{hi:hifs} and \ref{hifs:dhi}, we use the theoretical results to design two history independent
file systems.

%\input{informal}
%\input{contributions}

%We posit that for regulatory compliance, truly irrecoverable deletion can be
%achieved by utilizing history independent data structures to organize data
%(Section \ref{hitheory:hi}).
%To design such systems we will make use of history independent data structures
%Our goal is to design data structures and complete systems with history independent 
%characteristics.
%(and to prove such characteristics)

%-----------------------------------------------------------------------------

\section{A Quick Informal Look at HI}
\label{hitheory:informal}

History independence (HI) is concerned with the historical information
preserved within data structure states.
The preserved history may be illicitly used by adversaries to
violate regulatory compliance.
For example, an adversary may breach data retention laws by
recovering deleted data. 
Therefore, to understand history independence, we need to specify what we mean by
state, what we mean by history, and what an adversary can do.

\begin{figure}[t]
\begin{center}
\vspace{0.5cm}
\includegraphics[width=3.5in]{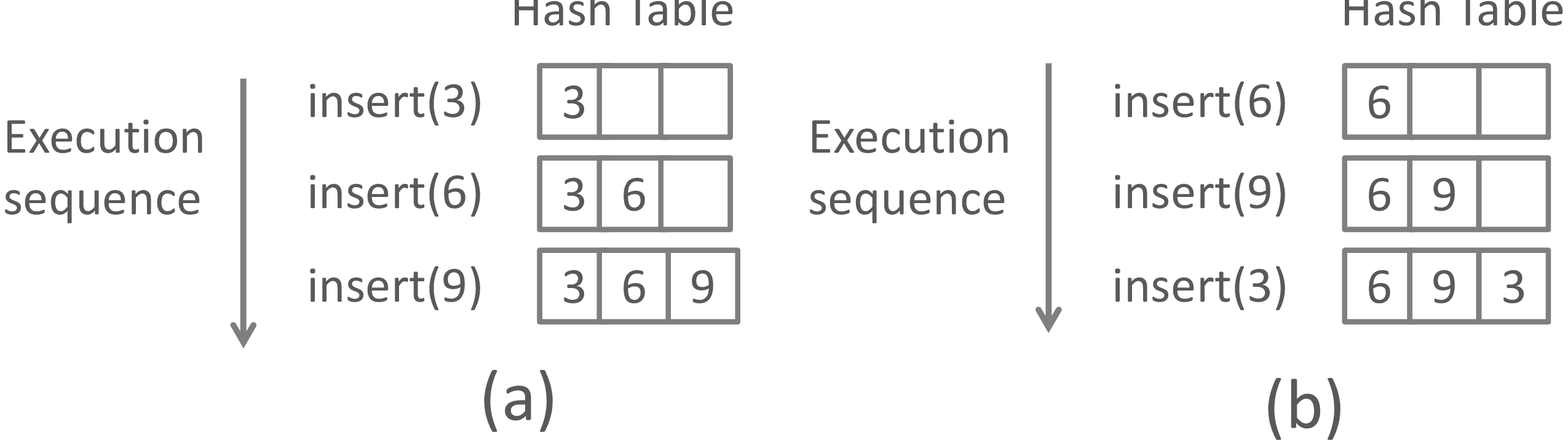}
\vspace{-3.8cm}
\caption[Example of a history dependent hash table]
{A history dependent hash table organizes the same data set
differently depending on the sequence of operations (i.e., history).
In this example, the hash table uses linear probing \cite{Mehta:2004:HDS:1044879}.
The number of hash table buckets is 3 and the hash function is modulo 3.
\label{fig:hitheory:hashtableexample}}
\vspace{-0.4cm}
\end{center}
\end{figure}
%

%In this section we gently introduce history independence, and later
%in Sections \ref{hitheory:preliminaries} to \ref{hitheory:theory2practice} we
%formalize all relevant concepts.

\smallskip
\noindent
\textbf{What is state?}\\
A data structure's state is an organization of data
%its entire content
on a physical medium such as memory or disk.

\smallskip
\noindent
\textbf{What is history?}\\
History is the sequence of operations that led to the current data structure state.
%constitute the state's history.

\smallskip
\noindent
\textbf{What is the threat?}\\
For many existing data structures, the current state
is a function of both data and history \cite{golovinthesis}.
Hence, 
%Consider an adversary with access to the current data structure
%state. If the current state is a function of history
%information about history
by analyzing the current state an adversary can derive
the state's history. Depending on the application the historical
information includes the following:
\begin{itemize*}
\item	Evidence of past existence of deleted data \cite{ficklebase}.
\item	The order in which votes were cast in a voting application \cite{naorcuckoo,blellochhashing}.
%\item	The order in which files were created (or written) to a file system (Chapter \ref{chapter:hifs}).
\item	Intermediate versions of published documents \cite{naorcuckoo}.
\end{itemize*}

To illustrate, consider the sample hash table data structure
of Figure \ref{fig:hitheory:hashtableexample}.
The sample hash table organizes the same data set differently
depending on the sequence of operations used.
Hence, an adversary that looks at the system memory can
potentially detect which operation sequence was used to get to the current
hash table state.

%If this hash table was used to store votes, then it would
%reveal the partial information about the order in which 
%the votes were cast, and also if any votes were 
%deleted.

\smallskip
\noindent
\textbf{What is history independence?}\\
History independence is a characteristic of a data structure.
A data structure is said to be history independent if from the adversary's
point of view, the current data structure state 
is a function of data only and not of history.
%is independent of the sequence of past operations.
Thus, the current state of a history independent
data structure reveals no information to the adversary about its history
other than what is inherently visible from the data itself.
We emphasize that history independence is concerned with historical information
that is revealed from data organization and not from the data. 

\smallskip

\noindent
%\textbf{Notions of history independence}\\
\textbf{Are there different kinds of history independence?}\\
%History independence was first introduced in \cite{23tree} under the banner of
%oblivious data structures. Later
Naor \etal \cite{naorantipersistence} introduced two
notions of history independence -- weak history independence (WHI) and strong
history independence (SHI). 
%Both are defined in Sections \ref{hi:weakhi} and
%\ref{hi:stronghi} respectively.

WHI and SHI differ in the number of data structure states
an adversary is permitted to observe.
Under WHI, an adversary is permitted to observe only the current
data structure state. For example, as in case of a stolen laptop.
Under SHI, an adversary is permitted several observations
of data structure states throughout a sequence of operations.
For example, as in case of an insider adversary who can obtain a periodic memory dump.
For SHI, the adversary should be unable to identify which 
sequence of operations was applied between any two adjacent observations.
%than what is inherent from the corresponding ADT state.
%Subsequent work has designed several data structures with history independent
%characteristics for the RAM model (Section \ref{hifs:related}).
%Their applications include incremental signature schemes \cite{naorcuckoo}, privacy
%in voting systems \cite{naorcuckoo,blellochhashing,molnarprom,moranwriteonce}, performing
%updates without revealing intermediate states \cite{naorantipersistence},
%debugging parallel computations \cite{blellochhashing}, reconciliation of
%dynamic sets \cite{naorcuckoo}, and un-traceable deletion \cite{ficklebase}. 
\smallskip
\subsection{Our Contributions}
\label{hitheory:intro:contributions}

%Weak history independence (WHI) and strong history independence (SHI)
%represent the two extreme ends of history independence.
%are the two existing notions of history independence.
WHI assumes a weak adversary while SHI is a very
powerful notion of history independence,
%requires canonical representations (defined in Section \ref{hi:canonical}), and is thereby
secure even against a computationally unbounded adversary \cite{golovinthesis}.
Currently, applications are restricted to using data structures with either WHI or SHI characteristics.
However, applications that do not fit into either WHI or SHI do exist.
For example, a journaling system that reveals no historical information other
than the last $k$ operations\footnote{We give additional examples in Section \ref{hi:deltahi}.}.
%Therefore, it is essential to understand the entire spectrum of history independence
%other than just WHI and SHI.
%Moreover, applications are currently restricted to using data structures with either WHI or SHI characteristics.
Further, WHI does not protect against insider adversaries and SHI results in inefficiency \cite{buchbinderbounds}.
Hence, there is a necessity for new notions of history independence targeted towards
specific application scenarios.
%are proposed, they can result in data structures that not
%only provide greater functionality but are also significantly more efficient. 
%The need for new notions of history independence other that WHI and SHI
%has been suggested \cite{naorantipersistence,golovinthesis}.
%However, a proper theoretical framework to define such notions of
%history independence has not yet been
%proposed.

In this paper we take the first steps towards better understanding
the history independence spectrum and its applicability to systems.
The contributions of this paper are:
\begin{itemize}
\item	The exploration of abstract data types, data structures,
	machine models,
	and memory representations (Section \ref{hitheory:preliminaries}).
	This is an essential step towards formalizing history independence. 
\item	New game-based definitions of weak and strong history
	independence (Sections \ref{hi:weakhi} and \ref{hi:stronghi})
	that are more appropriate for the security community
	as compared to existing terminology \cite{naorantipersistence,hartlinecharacterizinghi}.
\item	A new notion of history independence termed
	$\Delta$ history independence ($\Delta$HI).
	$\Delta$HI centers around a generic game-based definition
	of history independence and is malleable enough to accommodate WHI, SHI, and 
	a broad spectrum of
	new history independence notions (Section \ref{hi:deltahi}).
	%, this framework
	%also captures the existing notions of weak and strong history independence.
	In addition, $\Delta$HI helps to quantify the history revealed or hidden
	by existing data structures most of which have been designed without
	history independence in mind.
%\item	%A clear identification of scenarios where canonical representations are
	%necessary to achieve history independence (Table \ref{hitheory:table:hicanrep}).
	%over all possible combinations of types of programs,
	%secrecy of random bits, and adversarial computational abilities ().
	%We also prove that canonical representations are a must for history
	%independence in case of deterministic machine programs (Section \ref{hi:canonical}).
\item	A general recipe for designing history independent systems
	%(Section \ref{hitheory:theory2practice})
	and the recipe's use
	in designing a history independent file system (Section \ref{hi:hifs}).
\item	The design and evaluation of delete agnostic file system (DAFS). In DAFS,
	we re-design the file system layer to support new history independence
	notions. DAFS also increases file system resilience via journaling in the
	presence of history independence.
\end{itemize}

\begin{comment}
\noindent
\textbf{How does history independence achieve regulatory compliance?}\\
The current state of a history independent data structure is a function
of current data only. Data that was deleted in the past leaves no
effect on the current state that an adversary can detect.
History independent data structures are therefore ideal to organize
data in compliance with data retention Regulations \cite{cfr240,pipeda,eu-drd}
that require truly irrecoverable data erasure.
\end{comment}

\section{Preliminaries}
\label{hitheory:preliminaries}

%Our goal is a formalization of history independent data structures.
Formalizing history independence requires an understanding of data structures.
A data structure itself can be viewed as an implementation of an abstract data type (ADT)
on a machine model \cite{golovinthesis}.
An abstract data type (ADT) is a specification of operations for data
organization while a machine model represents a physical computing machine.
%(e.g., RAM machine).

%We cannot do so until we have a formal definition of a data structure
%itself. Precisely defining a data structure requires two concepts,
%the abstract data type and the machine model.
%Hence, in this section we provide the definitions for 
%abstract data types, machine models and data structures (in this order).
%Later, in Section \ref{hitheory:hi} we formalize history independence.

%We will do more than just list the necessary definitions.
%Wherever applicable, we will also discuss the aspects of ADTs, of machine models and
%of data structures that are relevant to history independence (in both theory and practice).
In the following, we provide an overview of ADTs, data structures, machine models,
and memory representations as proposed in \cite{golovinthesis} that are relevant to history independence.
Then, in Section \ref{hitheory:hi} we formalize history independence.

%Before we detail, we briefly summarize the concepts:
%An abstract data type (ADT) is a specification of operations for data
%organization. A machine model represents a physical computing machine
%(e.g., RAM machine). A data structure is an implementation of an ADT
%in a given machine model.

\subsection{Abstract Data Type (ADT)}
\label{hitheory:preliminaries:adt}
%Data is an essential component in any computation. Computational efficiency
%thus often depends on the efficiency with which data can be stored and 
%accessed. For several decades, research has thus focused on the design of
%efficient data organization techniques.

The specification of data organization techniques is often done via abstract
data types. The key characteristic of an ADT is that it specifies operations independently of any
specific implementation.
%This helps to reason about the effects of various ADT operations,
%the use of the ADT, and to compare different ADTs.
%Several ways of formalizing ADTs exists; such as $\Sigma$-algebra \cite{23tree}, axiomatic
%semantics \cite{nistadt}, etc.
% and state transition graphs \cite{golovinthesis}.
We use the concept proposed by Golovin \etal \cite{golovinthesis},
wherein an ADT is considered as a set of states
together with a set of operations. Each operation maps the current state to a new state.
%Conceptualizing an ADT in this way offers a middle-ground between $\Sigma$-algebra 
%and axiomatic semantics.
%It separates out ADT states and operations and is hence more detailed than $\Sigma$-algebra.
%At the same time it is sufficiently generic and avoids the case-by-case definitions
%required when using axiomatic semantics.
%We define an ADT as follows.

\begin{definition}
\label{def:adt}
{\em Abstract Data Type (ADT)} \\
An ADT $\mathcal{A}$ is a pentuple $(\mathcal{S}, s_{\phi}, \mathcal{O}, \Gamma, \Psi)$, where $\mathcal{S}$ is a set of states; $s_{\phi} \in \mathcal{S}$ is the initial state; $\mathcal{O}$ is a set of operations; $\Gamma$ is a set of inputs; $\Psi$ is a set of outputs; and
each operation $o \in \mathcal{O}$ is a function
\footnote{For brevity, we model each ADT operation with an input and an output. ADT operations
may accept no inputs or produce no outputs. Hence, an ADT operation can also be
modeled as the following functions: $o : \mathcal{S} \rightarrow \mathcal{S}$,
$o : \mathcal{S} \rightarrow \mathcal{S} \times \Psi_{o}$, or
$o : \mathcal{S} \times \Gamma_{o} \rightarrow \mathcal{S}$.}
$o : \mathcal{S} \times \Gamma_{o} \rightarrow \mathcal{S} \times \Psi_{o}$, where
$\Gamma_{o} \subseteq \Gamma$ and $\Psi_{o} \subseteq \Psi$.
\end{definition}

The ADT is initialized to state $s_{\phi}$. When an operation $o \in \mathcal{O}$
with input $i \in \Gamma_{o}$ is applied to an ADT state $s_{1}$, the ADT outputs
$\tau \in \Psi_{o}$ and transitions to a state $s_{2}$. The transition from
$s_{1}$ to $s_{1}$ is denoted
as $o(s_{1},i) \rightarrow (s_{2},\tau)$.
\smallskip

%\subsubsection{State Transition Graph For ADT}
%-----------------------------------------------------------------------------

\noindent
\textbf{The necessity of ADTs.}
%\label{hitheory:preliminaries:adt:need}
History independence requires that from an adversary's
point of view, the current data structure state 
is a function of data only and not of history.
In the context of history independence, an ADT models the history
revealed by data only.
Since we view a data structure as an ADT implementation (Section \ref{hitheory:preliminaries:ds}),
the ADT helps to clearly identify what the data structure
is permitted or not permitted to reveal about past operations.
Any history revealed by an ADT state can be revealed
by the corresponding data structure state.
Any history hidden by an ADT state must be hidden by the
corresponding data structure state.

\smallskip
\noindent
\textbf{ADT as a graph.} 
%\label{hitheory:preliminaries:adt:graph}
%
%\medskip
%\noindent
%\textbf{ADT as a graph: }
We can imagine the ADT to be a directed graph $\mathcal{G}$, where
each vertex represents an ADT state and each edge is labeled with
an ADT operation along with an ADT input and an ADT output.
The label for an edge between two vertices represents the operation that
causes the transition between the corresponding states.
%A sequence of operations is thus a path in the graph $\mathcal{G}$.
We call the graph $\mathcal{G}$, the state transition graph of the ADT.

Viewing an ADT as a graph will be particularly useful when we take
a deeper look into history independence in Section \ref{hitheory:preliminaries:hisemi}.

\subsection{Models of Execution}
\label{hitheory:preliminaries:execmodels}

An ADT is only a specification of operations for organizing data.
For more practical use, such as for efficiency analysis, concrete implementations of the
ADT operations are required. ADT implementations are provided via programs that can
be executed on a given machine model.
An ADT's implementation in a given machine model is a data structure
(Section \ref{hitheory:preliminaries:ds}). \\

\noindent
\textbf{RAM Model of Execution.}
\label{hitheory:preliminaries:rammodel}
The RAM model of execution models a traditional serial computer. The model
consists of two components, a central processing unit (CPU) and a 
random access memory (RAM). Both the CPU and RAM are finite state machines (FSM) \cite{modelsbook}.

The RAM consists of $m$ = $2^{u}$ storage locations. Each location is a $b$-bit word and
has a unique $\log_2 m$ bit address associated with it\footnote{This a bounded-memory RAM.}.
%<more RAM details ... instructions>
Two operations are permitted on a storage location in the RAM. First, a load
operation to access the $b$-bit bit word stored at the location. Second,
a store operation that copies a given $b$-bit word to the location. Typically,
the $b$-bit words are copied to or copied from CPU registers.

The CPU consists of $n$ $b$-bit registers and operates on a fetch-and-execute cycle \cite{modelsbook}.
The CPU has an associated set of instructions that it can perform. CPU instructions are
specified in a programming language.
%(e.g., assembly language).
A program in a RAM model is a finite sequence of programming language instructions.

A machine model can itself be considered as an ADT \cite{golovinthesis}. 
In this case, the set of ADT states is the set of all machine states, and the set of ADT
operations is the set of all machine programs. For the RAM model, the set of ADT states,
the set of inputs, and the set of outputs are all represented as bit strings.

\begin{definition}
\label{def:rammodel}
{\em Bounded RAM Machine Model} \\
A bounded RAM machine model $\mathcal{M}$ with $m$ $b$-bit memory words and $n$ $b$-bit CPU registers is a pentuple $(\mathcal{S}, s_{\phi}, \mathcal{P}, \Gamma, \Psi)$, where $\mathcal{S} = \{0, 1\}^{b(m+n)}$ is the set of machine states; $s_{\phi} \in \mathcal{S}$ is the initial state; $\mathcal{P}$ is the set of all programs of $\mathcal{M}$; $\Gamma = \{0, 1\}^{*}$ is a set of inputs; $\Psi = \{0, 1\}^{*}$ is a set of outputs; and each program $p \in \mathcal{P}$ is a function $p : \mathcal{S} \times \Gamma_{p} \rightarrow \mathcal{S} \times \Psi_{p}$,
where $\Gamma_{p} \subseteq \Gamma$ and $\Psi_{p} \subseteq \Psi$.
\end{definition}

$\mathcal{M}$ is initialized to state $s_{\phi}$.
If a program $p \in \mathcal{P}$ with input $i \in \Gamma_{p}$ is executed by the CPU when $\mathcal{M}$ is in state $s_{1}$, $\mathcal{M}$ outputs $\tau \in \Psi_{p}$ and transitions to a state $s_{2}$.
The transition from $s_{1}$ to $s_{2}$ is denoted as $p(s_{1},i) \rightarrow (s_{2},\tau)$.

\subsection{Data Structure}
\label{hitheory:preliminaries:ds}

%In the previous section, we hinted that a data structure
%is an ADT's implementation in a specific machine model.
%Now that we have defined both ADT and the RAM machine
%model we can formalize the data structure.

%An ADT is an abstract specification.
%For use in practice, a specific implementation of the ADT is required in
%a particular machine model. Such an implementation of an ADT is
%referred to as a data structure.

An implementation for an ADT in a given machine model
%, that is a data structure,
is obtained as follows.
\begin{itemize}
\item	A machine representation is chosen for each ADT input and output.
\item	For each ADT operation a machine program is selected that provides
	the functionality desired from the ADT operation.
\item	A unique machine state	is selected to represent the initial ADT state.
\end{itemize}

\noindent
We encapsulate the above steps in the following data structure definition.

%Since we are concerned only with the RAM model of execution we only consider
%data structures under this model.

\begin{definition}
\label{def:ds}
{\em Data Structure} \\
A data structure implementation of an ADT $\mathcal{A}$ in a bounded RAM machine model $\mathcal{M}$
is a quadruple $(\alpha, \beta, \gamma, s_{0}^{\mathcal{M}})$, where
$\mathcal{A} = (\mathcal{S}, s_{\phi}, \mathcal{O}, \Gamma, \Psi)$ as per definition \ref{def:adt},
$\mathcal{M} = (\mathcal{S}^{\mathcal{M}}, s_{\phi}^{\mathcal{M}}, \mathcal{P}^{\mathcal{M}}, \Gamma^{\mathcal{M}}, \Psi^{\mathcal{M}})$ as per definition \ref{def:rammodel},
$\alpha : \Gamma' \rightarrow \Gamma^{\mathcal{M}}$,
$\beta : \Psi' \rightarrow \Psi^{\mathcal{M}}$,
$\gamma : \mathcal{O} \rightarrow \mathcal{P}^{\mathcal{M}}$,
$s_{0}^{\mathcal{M}} \in \mathcal{S^{\mathcal{M}}}$,
$\Gamma' \subseteq \Gamma$ and
$\Psi' \subseteq \Psi$.
%
%Let $\mathcal{A} = (\mathcal{S}, s_{\phi}, \mathcal{O}, \Gamma, \Psi)$ be an ADT and $\mathcal{M} = (\mathcal{S}^{\mathcal{M}}, s_{\phi}^{\mathcal{M}}, \mathcal{P}^{\mathcal{M}}, \Gamma^{\mathcal{M}}, \Psi^{\mathcal{M}})$ be a bounded RAM machine model, as per definitions \ref{def:adt} and \ref{def:rammodel} respectively. Then, an implementation $\mathcal{D}$ of $\mathcal{A}$ in $\mathcal{M}$ (i.e., data structure in RAM model) is a quadruple $(s_{0}^{\mathcal{M}}, \alpha, \beta, \gamma)$ where, $s_{0}^{\mathcal{M}} \in \mathcal{S^{\mathcal{M}}}$, $\alpha : \Gamma' \rightarrow \Gamma^{\mathcal{M}}$, $\beta : \Psi' \rightarrow \Psi^{\mathcal{M}}$, and $\gamma : \mathcal{O} \rightarrow \mathcal{P}^{\mathcal{M}}$, such that, $\Gamma' \subseteq \Gamma$ and $\Psi' \subseteq \Psi$.
\end{definition}

$\alpha$ is a mapping from ADT inputs to machine inputs. That is, for any ADT input
$i$, $\alpha(i)$ is the machine representation of the input.
%For the RAM model this representation is a bit string.
Similarly, $\beta$ is the mapping from ADT outputs to machine outputs.
%Due to the finite number of available bits in the RAM machine, we may not be able
%to represent all the ADT inputs and outputs in our implementation.
%Hence, $\alpha$ and $\beta$ map only a subset of the ADT inputs (and outputs)
%to machine inputs (and outputs). We will come back
%to this later is Section \ref{hitheory:theory2practice:infinitestatespace},
%where we also discuss the case of infinite ADT states.
$\gamma$ is the mapping from ADT operations to machine programs.
For an ADT operation $o$,
$\gamma(o)$ is the machine program implementing $o$.
Finally, just as the ADT $\mathcal{A}$ is initialized to a unique state $s_{\phi}$, 
a unique machine state $s_{0}^{\mathcal{M}}$ is selected to represent
the initial data structure state.
\smallskip
\noindent
\textbf{Data Structure State.} 
A data structure state is a machine state.
%The set of all
%data structure states is the set of vertices in the data structure's
%state transition graph.
The set of all data structure states consists of all
machine states that are reachable from the initial data structure state 
via execution of machine programs implementing the ADT operations.
%----------------------------------------------------------------------------------------------

\smallskip
\noindent
\textbf{State Transition Graph For Data Structure.}
\label{hitheory:preliminaries:ds:graph}
A data structure can be considered to be a directed graph $\mathcal{G}$, where
each vertex represents a data structure state and each edge is labeled with
a machine program implementing an ADT operation along with a machine input
and a machine output.
The label for an edge between two vertices represents the program that
causes the transition between the corresponding states.
We call the graph $\mathcal{G}$, the state transition graph of the data structure.

\subsection{A Semi-Formal Look At HI}
\label{hitheory:preliminaries:hisemi}

\textbf{The non-isomorphism problem.}
In Section \ref{hitheory:informal} we introduced the two existing
history independence notions -- weak history independence (WHI)
and strong history independence (SHI)\footnote{Both WHI and SHI are
formalized in Section \ref{hitheory:hi}.}.

Non-isomorphism between the state transition graph of an ADT and of
its data structure implementation breaks SHI.
WHI on the other hand can be achieved even when the ADT and data structure
state transition graphs are non-isomorphic. First, we look at how
non-isomorphism breaks SHI and then we discuss how to achieve
WHI in the presence of non-isomorphism.

\begin{table}[t]
\vspace{0.6cm}
%\begin{center}
\setlength{\tabcolsep}{1pt}
\caption[]{
\begin{footnotesize}
Sample paths from ADT and data structure state transition graphs                                                                                                                   .\label{hifs:table:paths}
\end{footnotesize}
}
{%

\begin{tabular}{|c|c|}

\hline
{\bf Path}	& {\bf From}\\
		& {\bf Figure}\\
\hline

$p_{\mathcal{A}} = s_{\phi} \rightarrow \{1\} \rightarrow \{1,3\} \rightarrow \{1,3,6\}$ &
\ref{fig:hitheory:graphs}(a)\\
\hline

$p'_{\mathcal{A}} = s_{\phi} \rightarrow \{1\} \rightarrow \{1,6\} \rightarrow \{1,3,6\}$ &
\ref{fig:hitheory:graphs}(a)\\
\hline
$p_{\mathcal{D}} = s^{\mathcal{M}}_{\phi} \rightarrow <<\_,1,\_>> \rightarrow <<3,1,\_>> \rightarrow <<3,1,6>>$ &
\ref{fig:hitheory:graphs}(b)\\
\hline
$p'_{\mathcal{D}} = s^{\mathcal{M}}_{\phi} \rightarrow <<\_,1,\_>> \rightarrow <<6,1,\_>> \rightarrow <<6,1,3>>$ &
\ref{fig:hitheory:graphs}(b)\\
\hline

\end{tabular}}

%\vspace{-0.2cm}
%\caption[Sample paths from ADT and data structure state transition graphs]
%{Sample paths from ADT and data structure state transition graphs of Figure \ref{fig:hitheory:graphs}.
%\label{hifs:table:paths}}
\vspace{-0.4cm}
%\end{center}
\end{table}

\begin{figure*}[th]
\begin{center}
\vspace{0.9 cm}
\subfigure{
\hspace{-0.4cm}
\includegraphics[scale=0.26]{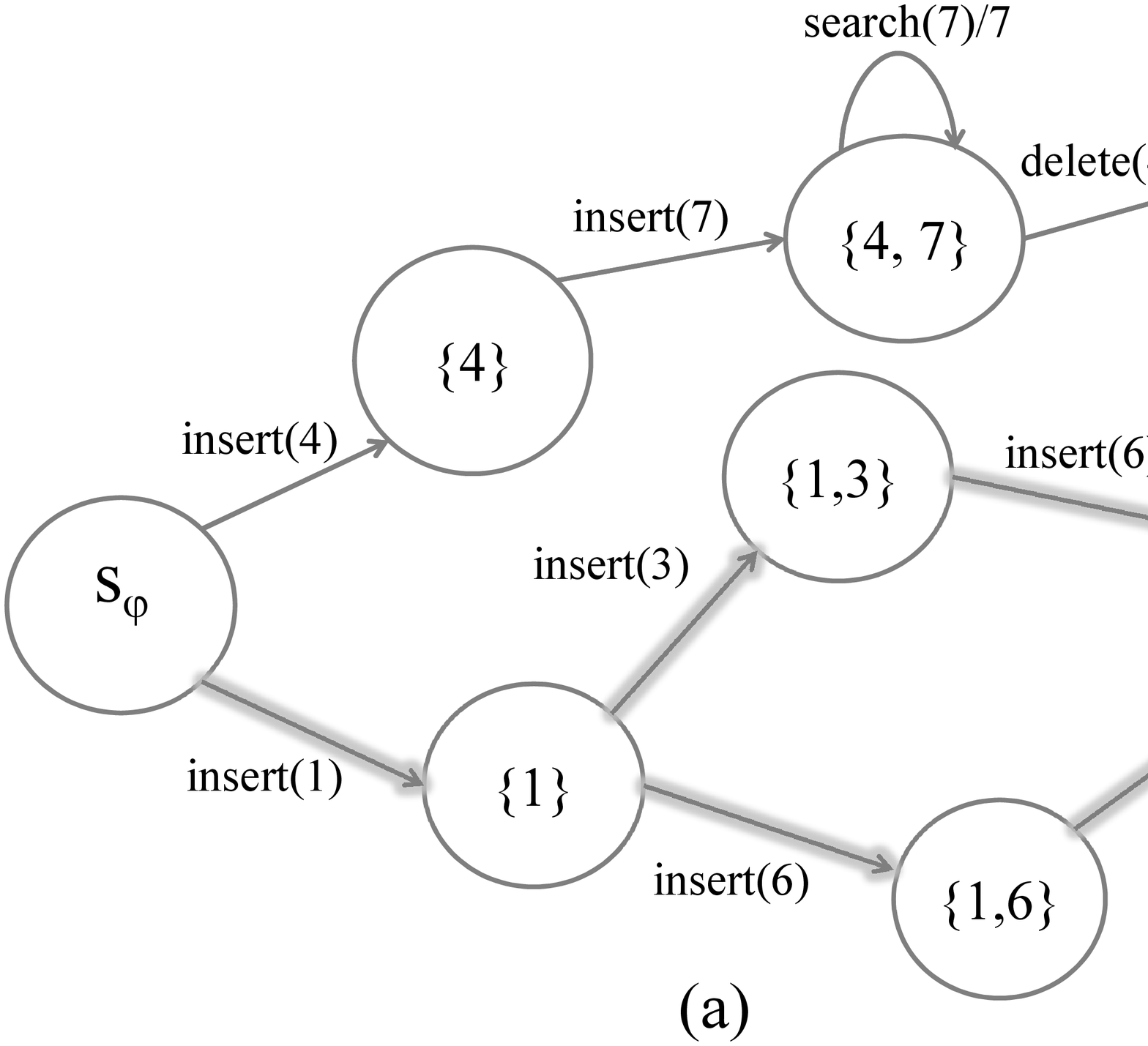}
}
\hspace{-0.2cm}
\subfigure{
\includegraphics[scale=0.26]{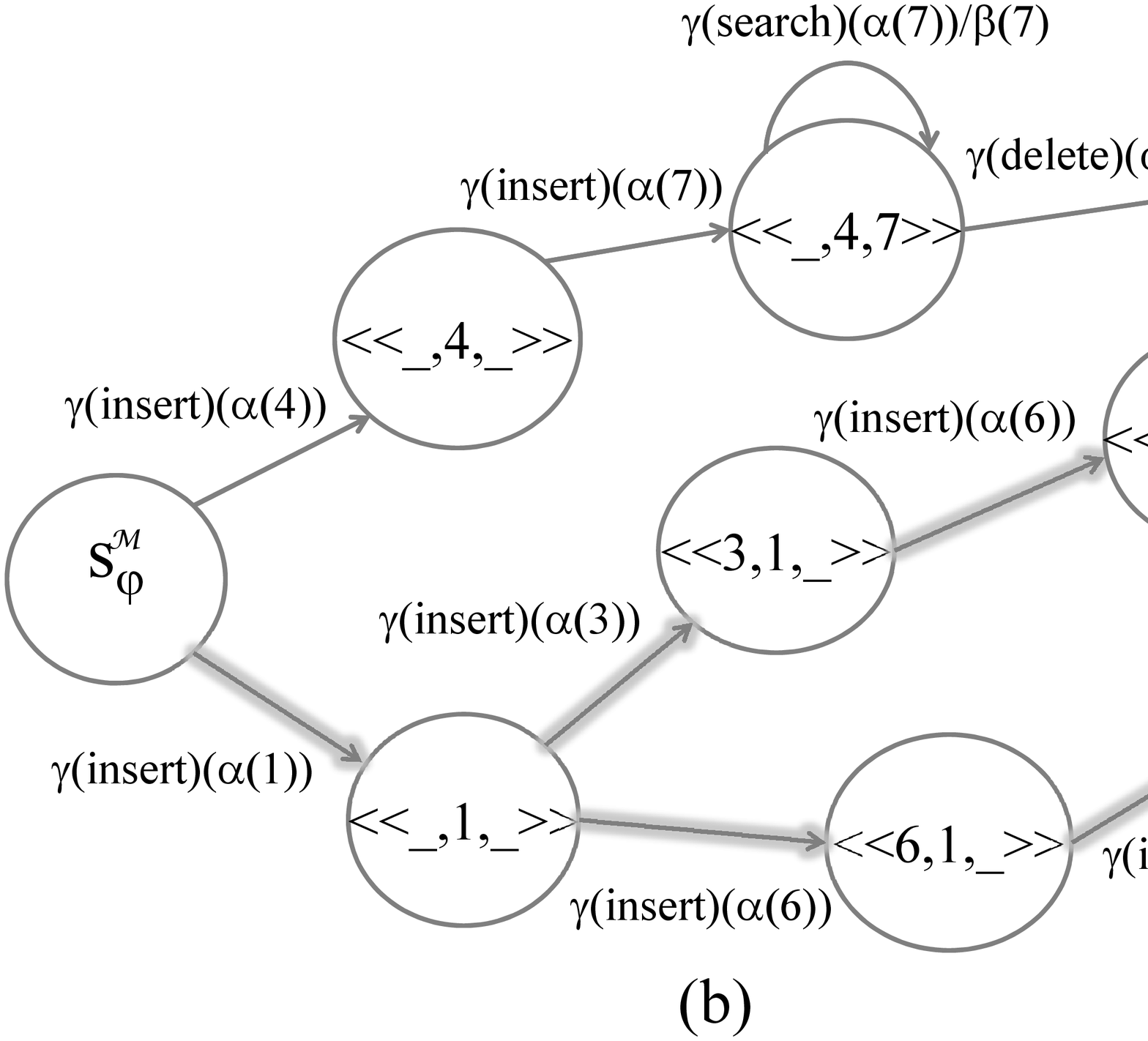}
}
\vspace{-2.0cm}
\caption[State transition graphs]{
\begin{footnotesize}
Example of non-isomorphism between ADT and data structure state transition graphs.
(a) Partial state transition graph for sample hash table ADT.
(b) Partial state transition graph for sample array-based hash table data structure
implementation using linear probing.
Number of hash table buckets is 3 and the hash function is $h(key) = key\;\%\;3$.
$\gamma$(insert), $\gamma$(search) and $\gamma$(delete) denote the machine programs
implementing the ADT operations insert, search and delete, respectively.
$o(i)/t$ denotes that ADT operation $o$ takes input $i$ and produces output $t$.
Similarly, $\gamma(o)(\alpha(i))/\beta(t)$ denotes that program $\gamma(o)$ takes input $\alpha(i)$
and produces output $\beta(t)$.
$\alpha(i)$ and $\beta(t)$ are the machine representations of the ADT input $i$ and ADT
output $t$, respectively.
Note that the vertices in figure (b) represent data structure states.
In the RAM model these will be bit strings. However, to convey data semantics we denote the
hash table array as $<<a_{0}, a_{1}, a_{2}>>$, where $a_{0}$, $a_{1}$, and $a_{2}$
are elements at buckets $0$ , $1$ and $2$, respectively. Underscore denotes an empty bucket.
Highlighted paths are referenced in Table \ref{hifs:table:paths},
%and \ref{hifs:table:sequences},
and in Section \ref{hitheory:preliminaries:hisemi}\\\\.
\end{footnotesize}
\label{fig:hitheory:graphs}}
\vspace{-1.0cm}
\end{center}
\end{figure*}

\smallskip
\noindent
\textbf{Why non-isomorphism breaks SHI?}
Non-isomorphism and thus the need for SHI
arises when an ADT state has multiple memory representations
\footnote{Many existing data structures have this property and are
hence, not history independent. Common
examples include the linked list, hash tables and B-Trees. In these data
structures different insertion order of the same set of data elements (i.e.,
the same ADT state)
results in different memory representations.}.
We will precisely define memory representations for ADT states
in Section \ref{hitheory:preliminaries:memrep}.
For now, it suffices to say the following:
A memory representation for an ADT state that is
reachable from the initial ADT state via a sequence of ADT operations is
the machine state reachable from the initial data structure state via
the corresponding program sequence.
For example, in Figure \ref{fig:hitheory:graphs}, the data structure states $<<3,1,6>>$ and
$<<6,1,3>>$ are memory representations of the ADT state $\{1,3,6\}$.

%Consider any two paths in the ADT state transition graph that
%have the same starting state and the same ending state. 
%Each path represents an ADT operation sequence, that is history. 
%By looking at the start and end states alone an adversary cannot
%deduce which path could have been used to go from the starting
%state to the ending state. Hence the ADT state transition graph
%alone does not reveal history.
%
%Now, each of the two paths in the ADT state transition graph
%has a corresponding path in the data structure state transition
%graph. The paths in the data structure state transition 0
%
%However, if the ending ADT state
%has multiple memory representations then by observing the
%ending

To illustrate how non-isomorphism breaks SHI, consider the example graphs
from Figure \ref{fig:hitheory:graphs}, example paths
from Table \ref{hifs:table:paths},
%and \ref{hifs:table:sequences},
% to understand how the non-isomorphism between ADT and data structure state transition
%graphs breaks history independence.
and an adversary with access to the initial ADT state $s_{\phi}$,
the initial data structure state $s^{\mathcal{M}}_{\phi}$,
the current ADT state $\{1,3,6\}$, and
the current data structure state which is either $<<3,1,6>>$ or $<<6,1,3>>$.
%\begin{itemize*}
%\item	The initial ADT state $s_{\phi}$.
%\item	The initial data structure state $s^{\mathcal{M}}_{\phi}$.
%\item	The current ADT state $\{1,3,6\}$.
%\item	The current data structure state (either $<<3,6,1>>$ or $<<6,3,1>>$).
%	The data structure states $<<3,6,1>>$ or $<<6,3,1>>$ are memory representations
%	of the ADT state $\{1,3,6\}$.
%\end{itemize*}

By looking at the ADT states alone, the adversary cannot determine which
sequence of ADT operations
%(path $p_{\mathcal{A}}$ or path $p'_{\mathcal{A}}$)
was used to arrive at the current ADT state $\{1,3,6\}$.
This is because there are two paths $p_{\mathcal{A}}$ and $p'_{\mathcal{A}}$
between the ADT states $s_{\phi}$ and $\{1,3,6\}$. Moreover, the ADT states
carry no information about the exact path used to transition from
$s_{\phi}$ to $\{1,3,6\}$.
Hence, the data alone gives the adversary no advantage in guessing
which sequence of ADT operations was applied in the past.

Now, by looking at the current data structure state,
the adversary can clearly identify which sequence of machine programs
%(path $p_{\mathcal{D}}$ or path $p'_{\mathcal{D}}$)
was used to arrive at the current data structure state.
The current data structure state is either $<<3,1,6>>$ or $<<6,1,3>>$.
There is a unique path from initial data structure state $s^{\mathcal{M}}_{\phi}$
to each of the states $<<3,1,6>>$ and $<<6,1,3>>$.
Hence, by observing the current data structure state, the adversary can identity whether
path $p_{\mathcal{D}}$ or path $p'_{\mathcal{D}}$ was used to transition
from state $s^{\mathcal{M}}_{\phi}$ to the current data structure state. 
Identification of the path in the data structure state transition graph
informs the adversary of the program sequence used.
Knowledge of the program sequence used in-turn tells the adversary
the sequence of ADT operations used. In conclusion, the data structure
implementation gives the adversary an advantage in guessing the history
of past execution, thereby breaking history independence.

%\noindent
%\textbf{Memory representations: }
%The data structure states $[3,6,1]$ or $[6,3,1]$ are both referred to as
%memory representations of the ADT state $\{1,3,6\}$.
%, since the ADT
%state can be constructed from both $[3,6,1]$ and $[6,3,1]$.
%-----------------------------------------------------------------------------

\medskip
\noindent
\textbf{How can we achieve history independence?}
The two known ways to make data history independent:

%Using the non-isomorphism between ADT and data structure state transition graphs
%Now that we have identified non-isomorphism as the reason for leakage of historical
%information by data structure implementations, we briefly discuss the known ways of
%achieving history independence (i.e., how to design history independent data structures).
%In Sections \ref{hi:randomization} and \ref{hi:canonical} we will provide more
%details.

\begin{enumerate}
\item	\emph{For SHI, make the ADT and the data structure state transition graphs isomorphic:}\\
	Data structures with state transition graphs isomorphic to their ADT's state
	transition graph are referred to as canonically represented  
	data structures. We discuss the necessity of canonical representations for
	SHI in Section \ref{hi:canonical}. SHI implies WHI.
\item	\emph{For WHI, make the data structure state transitions randomized:}\\
	Randomization here refers to the selection of the data structure state representing
	the corresponding ADT state. To illustrate, consider the example
	graphs from Figure \ref{fig:hitheory:graphs}.
	Both data structure states $<<3,1,6>>$ and $<<6,1,3>>$
	are valid memory representations of the ADT state $\{1,3,6\}$. For WHI, the choice
	between data structure states $<<3,1,6>>$ and $<<6,1,3>>$ to represent the ADT
	state $\{1,3,6\}$ must be random.
	
	\smallskip
	
	As shown in Figure \ref{fig:hitheory:graphrand}, randomization translates to addition of new
	paths in the data structure state transition graph
	to ensure the following:
	For any two ADT states $s_{0}$ and $s_{1}$, if there is a path
	in the ADT state transition graph between $s_{0}$ and $s_{1}$,
	then, there must be a path from all memory representations
	of ADT state $s_{0}$ to all memory representations of ADT state $s_{1}$ in the
	data structure's state transition graph.
	The choice of path in the data structure state transition graph
	between representations of ADT states $s_{0}$ and $s_{1}$
	is then made at random.
	%, and the
	%random bits are hidden from the adversary.
	
	\smallskip
	
	From the adversary's point of view randomization makes all memory representations
	of an ADT state equally likely to occur. Hence, observation
	of a specific representation gives the adversary no advantage in guessing
	the sequence of machine programs that led to the current data
	structure state. Since the adversary cannot identify the sequence
	of machine programs used, the adversary is also unable to identify the
	sequence of ADT operations that led to the current ADT state.
\end{enumerate}

\begin{figure}[th]
%\centerline{\includegraphics[scale=0.26]{figures/hashtabledsstgrand.eps}}
%\captionsetup[subfigure]{labelformat=empty}
\begin{center}
\vspace{1.5cm}
\includegraphics[scale=0.26]{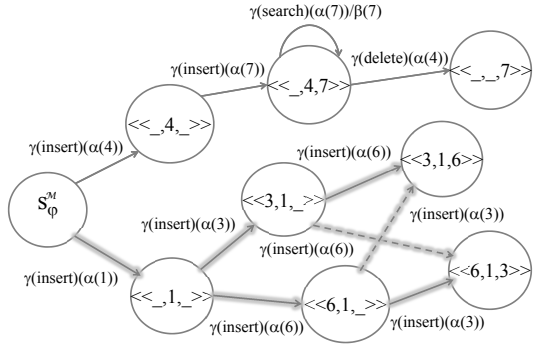}
\vspace{-2.4cm}
%\nocaption
\caption[Using randomization to achieve history independence]
{
\begin{footnotesize}
Using randomization to achieve history independence.
The dotted lines indicate new transitions added to the hash table data
structure state transition graph. Amongst all edges with the same
starting node and the same label, the choice of edge for state
transition is made at random.
\end{footnotesize}
}

\label{fig:hitheory:graphrand}
\vspace{-0.9cm}
\end{center}
\end{figure}

%Making the state transitions randomized achieves weak history independence (WHI),
%while strong history independence (SHI) mandates
%the ADT and the data structure state transition graphs to be isomorphic.
%More on this in Sections \ref{hi:randomization} and \ref{hi:canonical}.

\subsection{Memory Representations}
\label{hitheory:preliminaries:memrep}

%The last concept that remains to be formalized before we move on
%to formal definitions for history independence (Section \ref{hitheory:hi})
%is that of memory representations.

In the discussion of nonisomorphism and history independence above,
we informally introduced memory representations
for ADT states. We also showed that history
independence comes into picture when an ADT state has multiple
memory representations.
In short, the memory representation for an ADT state that is
reachable from the initial ADT state via a sequence of ADT operations, is
the machine state reachable from the initial data structure state via
the corresponding program sequence. We formally define memory representations
here and use them later in Section \ref{hitheory:hi} for the game-based definitions 
of history independence.

%An ADT is initialized to its initial state. From this point onwards each ADT operation
%causes the ADT to transition from the current state to a new state. Similarly,
%the ADT's implementation (i.e., the data structure) is initialized to a fixed
%machine state. From this point onwards, each machine program causes the
%data structure to transition from the current state to a new state.

%We can extend this notion of state transitions from individual operations
%and programs to a sequence of operations or programs.
%Each ADT operation has a corresponding machine program.
%%$p = \gamma(o)$. 
%Similarly, for a sequence of ADT operations $\delta$, there exists a corresponding
%sequence of machine programs $\delta^{\mathcal{M}}$, such that each program in 
%$\delta^{\mathcal{M}}$ implements the corresponding operation from $\delta$.
%The operation sequence $\delta$ when applied to the initial ADT state $s_{\phi}$
%will result in some ADT state $s_{n}$. Similarly, the corresponding program
%sequence  $\delta^{\mathcal{M}}$ when applied to the initial data structure
%state $s_{\phi}^{\mathcal{M}}$, will result in a machine state $s_{n}^{\mathcal{M}}$.
%Then, we say that the machine state $s_{n}^{\mathcal{M}}$ is a memory representation
%of the ADT state $s_{n}$. In short, the memory representation for an ADT state that is
%reachable from the initial ADT state via a sequence of ADT operations $\delta$, is
%the machine state reachable from the initial data structure state via
%the corresponding program sequence $\delta^{\mathcal{M}}$.

Let $\delta = \langle o_{1}, o_{2}, ..., o_{n} \rangle$ be a sequence of ADT operations and
$I = \langle i_{1}, i_{2}, ..., i_{n} \rangle$ be a sequence of ADT inputs.
We denote by $\mathbb{O}(\delta, s_{0}, I)$ the application of the ADT operation sequence
$\delta$ on ADT state $s_{0}$.
%, and define it as follows:

\noindent
\begin{footnotesize}
$\mathbb{O}(\delta, s_{0}, I) = \begin{cases} 
				    s_{0} & \mbox{if} |\delta| = 0\\
				    (s_{n},\tau_{n})|o_{k}(s_{k-1}, i_{k}) \rightarrow (s_{k}, \tau_{k}); \\ 1 \leq k \leq n & \mbox{otherwise} 
                               \end{cases}$
 
\end{footnotesize}

%\begin{center}
%$\Delta(s_{1}, s_{2}, \delta_{1}, \delta_{2}) =
%	\left\{ \begin{array}{cl}
%			1 & \mbox{if}\;\delta_{1}(s_{1}) \rightarrow s_{2}\;and\;\delta_{2}(s_{1}) \rightarrow s_{2}\\
%			0 & \mbox{otherwise}
%\end{array}\right.$
%\end{center}

If $\delta$ is empty no state transition occurs and no outputs are produced.
For nonempty sequence $\delta$, $s_{n}$ and $\tau_{n}$ denote the ADT state and the ADT output, respectively, 
produced by the final operation in sequence $\delta$.
%$\mathbb{O}(\delta, s_{0}, I)$ represents a path in the ADT's state
%transition graph (described in Section \ref{hitheory:preliminaries:adt})
%starting at state $s_{0}$ and ending at state $s_{n}$.

To summarize, we denote by $\mathbb{O}(\delta, s_{0}, I) \rightarrow (s_{n}, \tau_{n})$ that
the ADT operation sequence $\delta$ when applied to the ADT state $s_{0}$
with ADT input sequence $I$, results in the ADT state $s_{n}$ and ADT
output $\tau_{n}$.

Now, let $\delta^{\mathcal{M}} = \chi(\delta) = \langle \gamma(o_{1}), \gamma(o_{2}), ..., \gamma(o_{n}) \rangle$
be a sequence of machine programs corresponding to the ADT operation sequence $\delta$.
$\gamma(o_{k})$ is the machine program implementing the ADT operation $o_{k}$.
Then, we denote by  
$\mathbb{O}^{\mathcal{M}}(\delta^{\mathcal{M}}, s^{\mathcal{M}}_{0}, I)$ the application
of program sequence $\delta^{\mathcal{M}}$ on a machine state $s^{\mathcal{M}}_{0}$.
%, and define it as follows:

\noindent
\begin{footnotesize}
$\mathbb{O}^{\mathcal{M}}(\delta^{\mathcal{M}}, s^{\mathcal{M}}_{0}, I) = 
		\begin{cases}
					s^{\mathcal{M}}_{0} & \mbox{if}\;|\delta^{\mathcal{M}}| = 0\\
					(s^{\mathcal{M}}_{n}, \beta(\tau_{n})) \\ | \gamma(o_{k})(s^{\mathcal{M}}_{k-1}, \alpha(i_{k})) \rightarrow \{s^{\mathcal{M}}_{k}, \beta(\tau_{k})\}; \\ 1 \leq k \leq n & \mbox{otherwise}
		\end{cases}$
\end{footnotesize}
		
%\footnotetext{$(s^{\mathcal{M}}_{n}, \beta(\tau_{n}))| \gamma(o_{k})(s^{\mathcal{M}}_{k-1}, \alpha(i_{k})) \rightarrow \{s^{\mathcal{M}}_{k}, \beta(\tau_{k})\}; 1 \leq k \leq n$}		

%\begin{center}
%$\mathbb{O}(\delta^{\mathcal{M}}, s^{\mathcal{M}}_{0}, I) = \{ \{s^{\mathcal{M}}_{n}, \beta(\tau_{n})\} | \gamma(o_{k})(s^{\mathcal{M}}_{k-1}, \alpha(i_{k})) \rightarrow \{s^{\mathcal{M}}_{k}, \beta(\tau_{k})\}; 1 \leq k \leq n\}$.
%\end{center}

%$\delta^{\mathcal{M}}$ is empty when the ADT operation sequence $\delta$ is empty.
Here, $\alpha(i)$ and $\beta(\tau)$ denote the machine representations for an
ADT input $i$ and an ADT output $\tau$, respectively. $s^{\mathcal{M}}_{n}$ and
$\beta(\tau_{n})$ are the machine state and the machine
output, respectively, produced by the final program in sequence $\delta^{\mathcal{M}}$.

In summary, we denote by $\mathbb{O}^{\mathcal{M}}(\delta^{\mathcal{M}}, s^{\mathcal{M}}_{0}, I) \rightarrow (s^{\mathcal{M}}_{n}, \tau^{\mathcal{M}}_{n})$ that a program sequence $\delta^{\mathcal{M}}$ when applied to a machine state
$s^{\mathcal{M}}_{0}$ with an ADT input sequence $I$, results in a machine state
$s^{\mathcal{M}}_{n}$ and a machine output $\tau^{\mathcal{M}}_{n}$.
%$\mathbb{O}(\delta^{\mathcal{M}}, s^{\mathcal{M}}_{0}, I)$ represents a path
%in the data structure's state transition graph (described in Section \ref{hitheory:preliminaries:ds})
%starting at state $s^{\mathcal{M}}_{0}$ and ending at state $s^{\mathcal{M}}_{n}$.

%--------------------------------------------------------------
\begin{definition}
\label{def:mr}
{\em Memory Representations} \\
The set of memory representations of an ADT state $s$, denoted by $m(s)$,
is the set of data structure states, defined as

\noindent
\begin{footnotesize}
$m(s) = \begin{cases}
	
				 s^{\mathcal{M}}_{0}  & \mbox{if}\;s = s_{\phi}\\
				 s^{\mathcal{M}} \; | \; \mathbb{O}^{\mathcal{M}}(\delta_{k}^{\mathcal{M}}, s^{\mathcal{M}}_{0}, I_{k}) \rightarrow \{s^{\mathcal{M}}, \beta(\tau_{|\delta_{k}|})\}; \\ 1 \leq k \leq n  & \mbox{otherwise}
		
\end{cases}$
\end{footnotesize}
\smallskip
%\footnotetext{$s^{\mathcal{M}} \; | \; \mathbb{O}^{\mathcal{M}}(\delta_{k}^{\mathcal{M}}, s^{\mathcal{M}}_{0}, I_{k}) \rightarrow \{s^{\mathcal{M}}, \beta(\tau_{|\delta_{k}|})\};
%1 \leq k \leq n $}

where,
$s^{\mathcal{M}}_{0}$ is the initial data structure state;
$I_{1}, I_{2}, ..., I_{n}$ are sequences of ADT inputs;
$\delta_{1}, \delta_{2}, ..., \delta_{n}$ are ADT operation sequences,
each of which when applied to the initial ADT state $s_{\phi}$ results in state $s$,
that is $\mathbb{O}(\delta_{k}, s_{\phi}, I_{k}) \rightarrow (s, \tau_{k})$;
$\delta_{k}^{\mathcal{M}} = \chi(\delta_{k})$ denotes the program sequence corresponding to ADT operation sequence $\delta_{k}$;
$|I_{k}| = |\delta_{k}|$;
$1 \leq k \leq n$.
\end{definition}
%$\zeta$ is a set of ADT operation sequences such that,
%all sequences in $\zeta$ when applied to the initial ADT state $s_{\phi}$
%result in state $s$.
%That is,
%$\zeta = \{ \delta_{k}\;|\;\mathbb(\delta_{k}, s_{\phi}, I_{k}) = \{s, \tau_{k}\}$,
%$I_{k} = \{i_{1}, i_{2}, ..., i_{|\delta_{k}|}\;|\;i_{v} \in \Gamma; 1 \leq v \leq |\delta_{k}|\}$
%
%ADT $\mathcal{A} = (\mathcal{S}, s_{\phi}, \mathcal{O}, \Gamma, \Psi)$ (as per definition \ref{def:adt}),
%bounded RAM machine model $\mathcal{M} = (\mathcal{S}^{\mathcal{M}}, s_{\phi}^{\mathcal{M}}, \mathcal{P}^{\mathcal{M}}, \Gamma^{\mathcal{M}}, \Psi^{\mathcal{M}})$ (as per definition \ref{def:rammodel}), and
%data structure $\mathcal{D} = (\alpha, \beta, \gamma, s_{0}^{\mathcal{M}})$ (as per definition \ref{def:ds}).

Here $m$ is the mapping $m : \mathcal{S} \rightarrow 2^{\mathcal{S}^\mathcal{D}}$,
where $S$ is the set of all ADT states,
$\mathcal{S}^\mathcal{D}$ is the set of all data structure states, and
$2^{\mathcal{S}^\mathcal{D}}$ denotes the power set of $\mathcal{S}^\mathcal{D}$.

\subsubsection{Dealing With Infinite ADT State Space}
\label{hitheory:preliminaries:memrep:infinitestatespace}

The set of machine states for the bounded RAM
model is finite since there are finite number of available bits.
Hence, a data structure implementation on a bounded RAM
model can only have a finite number of data structure states.
The set of ADT states on the other hand can be infinite.
%For e.g., an ADT with set of 
For an ADT with infinite states, a data structure implementation
%(with finite number of states)
will be unable to uniquely represent all the ADT states.
The case of infinite ADT states is of particular importance for
canonically represented data structures that require the state
transition graphs of the ADT and of the data structure to be
isomorphic, that is, each ADT state has a unique memory representation.

We will look at canonical representations in detail within the
context of history independence in Section \ref{hi:canonical}.
Here, we list two work-arounds to dealing with infinite ADT state space.

\begin{enumerate}
\item	Redefine the ADT, such that the number of ADT states is less than or equal to
	the number of machine states.
\item	Design each machine program implementing an ADT operation, such that the program
	produces a special output when an ADT state cannot be represented using
	the available machine bits. For example, an out-of-memory error.
\end{enumerate}

%Let $\mathcal{A} = (\mathcal{S}, s_{\phi}, \mathcal{O}, \Gamma, \Psi)$ be an ADT and
%$\mathcal{M} = (\mathcal{S}^{\mathcal{M}}, s_{\phi}^{\mathcal{M}}, \mathcal{P}^{\mathcal{M}},
%\Gamma^{\mathcal{M}}, \Psi^{\mathcal{M}})$ be
%a bounded RAM machine model, as per definitions \ref{def:adt} and \ref{def:rammodel}
%respectively. Then, it is possible that $|\mathcal{S}| > |\mathcal{S}^{\mathcal{M}}|$.
%In fact, $\mathcal{S}$ can be a uncountable. $\mathcal{S}^{\mathcal{M}}$
%on the other hand is always a finite set since $\mathcal{M}$ is bounded-RAM machine model,
%with a finite number of available bits.
%
%Hence, in practice a data structure implementation for the ADT $\mathcal{A}$ on $\mathcal{M}$
%may not be able to represent all the ADT states.
%There are two work arounds to this problem listed below.
%------------------------------------------------------------------
%\subsubsection{Infinite Input and Output Space}
%\label{hitheory:theory2practice:infiniteio}
%
%Similar to infinite ADT states,
%The input and output sets of the ADT too can be infinite
%sets. Hence, all ADT inputs and outputs may not be
%representable using the finite machine bits.
%An ADT implementation therefore may only be able to operate over a 
%subset of the ADT input and output sets.
%This is captured in our data structure definition in 
%Section \ref{hitheory:preliminaries:ds}.

%--------------------------------------------------------------------------------------------------------

  \section{History Independence}
\label{hitheory:hi}

%In Section \ref{hitheory:intro:hi} we informally introduced history independence and
%briefly discussed the two existing notions of history independence --
%weak history independence (WHI) and strong history independence (SHI).
%We formalized the concepts of ADT (Section \ref{hitheory:preliminaries:adt}),
%bounded RAM machine model (Section \ref{hitheory:preliminaries:execmodels}),
%data structure (Section \ref{hitheory:preliminaries:ds}),
%and memory representations (Section \ref{hitheory:preliminaries:memrep}).
%In Section \ref{hitheory:preliminaries:hisemi} we illustrated how
%non-isomorphism between ADT and data structure state transition graphs breaks SHI.
%We also discussed the use of randomization to achieve WHI in the presence
%of non-isomorphism between ADT and data structure state transition graphs.

%(Section \ref{hitheory:preliminaries:hisemi}).
%canonical representations (Section \ref{hitheory:preliminaries:hisemi}).
%
%\medskip
%\noindent
%\textbf{Summary of Sections \ref{hitheory:intro} and \ref{hitheory:preliminaries}: }

Now that we are equipped with the necessary concepts
(ADT, RAM machine model, data structure, and memory representations),
we proceed to formalize history independence.
We give new game-based definitions for both WHI and SHI
(Sections \ref{hi:weakhi} and \ref{hi:weakhi}).
The new definitions are equivalent to existing proposals \cite{naorcuckoo,hartlinecharacterizinghi}
but more appropriate for the security community since they follow
the game-based construction of semantic security.
Further, our new definitions naturally extend to accommodate
other notions of history independence beyond WHI and SHI. 
\subsection{Weak History Independence (WHI)}
\label{hi:weakhi}

WHI was introduced for scenarios wherein an adversary observes only
the current data structure state. For example, as in the case of a stolen
laptop.
%The current data structure state is the memory representation
%of the current ADT state.
%WHI then requires that observation of the current data structure state 
%reveals no additional historical information to the adversary other than
%what is inherently available from the current ADT state.

Informally,
a data structure is said to be weakly history independent if for any two sequences
of ADT operations $\delta_{1}$ and $\delta_{2}$, that take the ADT from
initialization to a state $s$, observation of any memory representation
of state $s$ gives the adversary no advantage in guessing
whether sequence $\delta_{1}$ or $\delta_{2}$ was used
to get to $s$.

We define weak history independence (WHI) by the following game:
\begin{framed}
\begin{footnotesize}
%\medskip
Let $\mathcal{A} = (\mathcal{S}, s_{\phi}, \mathcal{O}, \Gamma, \Psi)$ be an ADT,
$\mathcal{M} = (\mathcal{S}^{\mathcal{M}}, s_{\phi}^{\mathcal{M}}, \mathcal{P}^{\mathcal{M}}, \Gamma^{\mathcal{M}}, \Psi^{\mathcal{M}})$ be a bounded RAM machine model, 
and $\mathcal{D} = (\alpha, \beta, \gamma, s_{0}^{\mathcal{M}})$ be a data structure implementing $\mathcal{A}$ in $\mathcal{M}$, as per definitions \ref{def:adt}, \ref{def:rammodel} and \ref{def:ds}, respectively.

\begin{enumerate}
\item	A probabilistic polynomial time-bounded adversary selects the following:
	%\begin{itemize}
	%\item
	An ADT state $s$; two sequences of ADT operations $\delta_{0}$ and $\delta_{1}$; and two sequences of ADT inputs $I_{0}$ and $I_{1}$;
	such that
			$\mathbb{O}(\delta_{0}, s_{\phi}, I_{0}) \rightarrow (s, \tau)$ and
			$\mathbb{O}(\delta_{1}, s_{\phi}, I_{1}) \rightarrow (s, \tau)$.
	Both $\delta_{1}$ and $\delta_{2}$ take the ADT from the initial state $s_{\phi}$
	to state $s$ producing the same output $\tau$.
	%\end{itemize}
\item	The adversary sends $s$, $\delta_{0}$, $\delta_{1}$, $I_{0}$ and $I_{1}$ to the challenger.
\item	The challenger flips a fair coin $c \in \{0,1\}$ and computes 
	$\mathbb{O}^{\mathcal{M}}(\delta_{c}^{\mathcal{M}}, s^{\mathcal{M}}_{0}, I_{c}) \rightarrow (s^{\mathcal{M}}, \tau^{\mathcal{M}})$,
	where $\delta_{c}^{\mathcal{M}} = \chi(\delta_{c})$ and $\tau^{\mathcal{M}} = \beta(\tau)$.
	That is, the challenger applies the program sequence $\delta_{c}^{\mathcal{M}}$
	corresponding to the ADT operation sequence $\delta_{c}$ to the data
	structure initialization state $s^{\mathcal{M}}_{0}$, resulting in a memory representation
	$s^{\mathcal{M}}$ of ADT state $s$ and a machine output $\tau^{\mathcal{M}}$.
\item	The challenger sends the memory representation $s^{\mathcal{M}}$ to the adversary.
\item	The adversary outputs $c' \in \{0,1\}$.
\end{enumerate}

\begin{center}
The adversary wins the game if $c' = c$.
\end{center}

$\mathcal{D}$ is said to be weakly history independent if the advantage of the adversary for winning the game, 
defined as $|Pr[c^{'}=c]-1/2|$ is negligible
(where ``negligible'' is defined over any implementation-specific security parameters of the programs in $\mathcal{P^M}$).\footnotemark
\end{footnotesize}

\end{framed}
\footnotetext{For example PRNG seeds when using randomization, or keys when using encryption.}
Since WHI permits the adversary to make a single observation,
% observe only one data structure (machine) state,
the adversary is allowed to choose the end state only in step 1.
The starting state for the chosen ADT operation sequences is
always the initial ADT state $s_{\phi}$.
%and the initial data structure state $s^{\mathcal{M}}_{0}$ respectively.
Recall from the data structure definition (Section \ref{hitheory:preliminaries:ds})
that the initial ADT state has a fixed memory representation, which is the initial
data structure state $s^{\mathcal{M}}_{0}$. 
Hence, in step 3, the challenger applies the adversary-selected sequence to the 
memory representation $s^{\mathcal{M}}_{0}$ of $s_{\phi}$.
%initial data structure state $s^{\mathcal{M}}_{0}$.
%This is captured in step 3 of the game above, where the challenger applies the
%selected sequence on the initial data structure state. 

%Winning the game means that
If the adversary is able to identify the
ADT operation sequence chosen by the challenger in step 3,
then the adversary wins
%\footnote{A function $f$ is negligible if for every positive polynomial poly(·) there
%exists an integer $n > 0$ such that for all $x > n$, $f(x) < \frac{1}{poly(x)}$.}
the game.
Winning the game implies the adversary was able to determine the operation sequence
that led to the current ADT state by observing the state's memory representation,
%current data structure state,
thereby breaking WHI.
%------------------------------------------------------------------------------------
\subsection{Strong History Independence (SHI)}
\label{hi:stronghi}

Unlike WHI, SHI is applicable when an adversary
can observe multiple memory representations throughout a sequence of operations
For example, as in case of an insider who can obtain a periodic memory dump. 
SHI requires that the adversary must not gain any additional information about the sequence
of operations between any two adjacent observations than what is inherently
available from the corresponding ADT states.

Informally, a data structure is said to be strongly history independent if
for any two sequences of ADT operations $\delta_{1}$ and $\delta_{2}$, that
take the ADT from a state $s_{1}$ to a state $s_{2}$, observations of
any memory representations of states $s_{1}$ and $s_{2}$
give the adversary no advantage in guessing
whether sequence $\delta_{1}$ or $\delta_{2}$ was used
to go from $s_{1}$ to $s_{2}$.

We define strong history independence (SHI) by the following game:

\begin{framed}
\begin{footnotesize}
\medskip
Let $\mathcal{A} = (\mathcal{S}, s_{\phi}, \mathcal{O}, \Gamma, \Psi)$ be an ADT, $\mathcal{M} = (\mathcal{S}^{\mathcal{M}}, s_{\phi}^{\mathcal{M}}, \mathcal{P}^{\mathcal{M}}, \Gamma^{\mathcal{M}}, \Psi^{\mathcal{M}})$ be a bounded RAM machine model, 
and $\mathcal{D} = (\alpha, \beta, \gamma, s_{0}^{\mathcal{M}})$ be a data structure implementing $\mathcal{A}$ in $\mathcal{M}$, as per definitions \ref{def:adt}, \ref{def:rammodel} and \ref{def:ds} respectively.

\begin{enumerate}
\item	A probabilistic polynomial time-bounded adversary selects the following.
	\begin{itemize}
	\item	Two ADT states $s_{1}$ and $s_{2}$;
		two sequences of ADT operations $\delta_{0}$ and $\delta_{1}$; and two sequences of ADT inputs $I_{0}$ and $I_{1}$;
		such that
				$\mathbb{O}(\delta_{0}, s_{1}, I_{0}) \rightarrow (s_{2}, \tau)$ and
				$\mathbb{O}(\delta_{1}, s_{1}, I_{1}) \rightarrow (s_{2}, \tau)$.
		Both $\delta_{1}$ and $\delta_{2}$ take the ADT from state $s_{1}$
		to state $s_{2}$ producing the same output $\tau$.
	\item 	A memory representation $s_{1}^{\mathcal{M}}$ of ADT state $s_{1}$.
	\end{itemize}
\item	The adversary sends $s_{1}$, $s_{1}^{\mathcal{M}}$, $\delta_{0}$, $\delta_{1}$, $I_{0}$ and $I_{1}$ to the challenger.
\item	The challenger flips a fair coin $c \in \{0,1\}$ and computes 
	$\mathbb{O}^{\mathcal{M}}(\delta_{c}^{\mathcal{M}}, s^{\mathcal{M}}_{1}, I_{c}) \rightarrow (s^{\mathcal{M}}_{2}, \tau^{\mathcal{M}})$,
	where $\delta_{c}^{\mathcal{M}} = \chi(\delta_{c})$ and $\tau^{\mathcal{M}} = \beta(\tau)$.
	That is, the challenger applies the program sequence $\delta_{c}^{\mathcal{M}}$
	corresponding to the ADT operation sequence $\delta_{c}$ to the data
	structure state $s^{\mathcal{M}}_{1}$, resulting in a memory representation
	$s^{\mathcal{M}}_{2}$ of state $s_{2}$ and a machine output $\tau^{\mathcal{M}}$.
\item	The challenger sends the memory representation $s^{\mathcal{M}}_{2}$ to the adversary.
\item	The adversary outputs $c' \in \{0,1\}$.
\end{enumerate}

\begin{center}
The adversary wins the game if $c' = c$.
\end{center}

$\mathcal{D}$ is said to be strongly history independent if the advantage of the adversary for winning the game, 
defined as $|Pr[c^{'}=c]-1/2|$ is negligible
(where ``negligible'' is defined over any implementation-specific security parameters of the programs in $\mathcal{P^M}$).
\end{footnotesize}
\end{framed}
Winning the game means that the adversary was able to determine the operation sequence
that took the ADT from state $s_{1}$ to state $s_{2}$, thereby breaking SHI.

SHI implies WHI. If the ADT state $s_{1}$ chosen by the adversary in step
1 is the initial ADT state $s_{\phi}$, then the SHI game reduces to the WHI game of
Section \ref{hi:weakhi}.
%------------------------------------------------------------------------------------
%\clearpage
\subsection{Equivalence to Existing Definitions}
\label{hitheory:equivalence}
WHI and SHI were first introduced by Naor \etal \cite{naorantipersistence}.
Later, Hartline \etal \cite{hartlinecharacterizinghi}
introduced new definitions for WHI and SHI. However, Hartline \etal showed that their definitions
although less complex are equivalent to the ones proposed by Naor \etal
%We show that our game-based definitions of WHI and SHI (Sections \ref{hi:weakhi} and \ref{hi:stronghi})
%are nearly equivalent to the definitions by Hartline \etal with a slight
%difference that is necessary for practicality.
Our game-based definitions of WHI and SHI (Sections \ref{hi:weakhi} and \ref{hi:stronghi})
differ slightly from the definitions by Hartline \etal
%The slight difference we introduce is necessary for practicality.
Specifically, Hartline \etal assume a computationally unbounded adversary.
We address history independence in the presence of computationally bounded adversaries
to be more in-line with reality.
Further, new definitions were necessary to overcome impreciseness
in existing definitions and to develop a framework for new history
independence notions beyond WHI and SHI. We detail in the following.

%formalized concepts that were imprecise

Hartline \etal defined weak history independence as follows. 
\begin{definition}
\label{def:hshi}
{\em Weak History Independence (WHI)} \\
A data structure implementation is weakly history independent if, for any two
sequences of operations X and Y that take the data structure from initialization
to state A, the distribution over memory after X is performed 
is identical to the distribution after Y. That is:
\end{definition}
\begin{center}
$(\phi \overset{X}{\longrightarrow} A) \land (\phi \overset{Y}{\longrightarrow} A)
\Longrightarrow
\forall$ \textbf{a} $\in A$,
\textbf{Pr\big[$\phi$ $\overset{X}{\longrightarrow}$ a\big] = Pr\big[$\phi$ $\overset{Y}{\longrightarrow}$ a\big]}
\end{center}
In the above definition, $\phi \overset{X}{\longrightarrow} B$ denotes that a operation
sequence $X$ when applied to the initial state $\phi$, results in state $A$.
The notation \textbf{a} $ \in A$ means than \textbf{a} is a memory representation of state $A$.
\textbf{Pr\big[$\phi$ $\overset{X}{\longrightarrow}$ a\big]}
denotes the probability that a sequence $X$ when applied to initial state $\phi$, results
in representation \textbf{a}.
%------------------------------------------------------------------------------------------------

\medskip
\noindent
\textbf{Reconciling terminology}
\smallskip

Hartline \etal do not formalize the concepts of data structure, data structure state
and memory representations. A data structure's state is referred to as the data structure's
content. Memory representation of a data structure state is the physical contents
of memory that represent that state. We note that Naor \etal also used the
same terminology in their definitions.

The WHI definition by Hartline \etal is imprecise in the following.
\begin{itemize}
\item	Operation inputs and outputs are not considered.
%	Considering operation outputs is particularly important since outputs
%	can break history independence even when the underlying memory
%	representations preserve history independence.
%
\item	The same operation sequences are considered applicable to both data structure
	states and to memory representations. The mechanisms for the 
	applicability are not specified.
\item	The connection between a data structure's state and the state's memory
	representations is not precisely specified.
\end{itemize}

Following Golovin \etal \cite{golovinthesis} we use the ADT concept to model logical
states (or content) and define a data structure as an ADT's implementation (Sections \ref{hitheory:preliminaries:adt} - \ref{hitheory:preliminaries:ds}).
A data structure state is therefore the memory representation of an ADT
state.
Separating ADT and data structure concepts enables us to precisely define
memory representations (Section \ref{hitheory:preliminaries:memrep})
for various machine models;
understand history independence from the perspective of state transition
graphs; and to build a framework for defining new history
independence notions other than SHI and WHI (Section \ref{hitheory:hi:generalizinghi}).

To summarize the differences in terminology, what Hartline \etal refer to as data structure state in definition \ref{def:hshi}
is an ADT state in our model. Further, we refer to a memory representation in definition \ref{def:hshi}
as a data structure state.

For WHI, Hartline \etal require a data structure implementation to satisfy the following:
\begin{center}\label{eq:hwhi}
$(\phi \overset{X}{\longrightarrow} A) \land (\phi \overset{Y}{\longrightarrow} A)
\Longrightarrow
\forall \textbf{a} \in A,
Pr\big[\phi \overset{X}{\longrightarrow} \textbf{a}\big] = Pr\big[\phi \overset{Y}{\longrightarrow} \textbf{a}\big]$
\end{center}

Our game-based definition of WHI poses the following slightly relaxed requirement:
\begin{center}\label{eq:whi}
$(\phi \overset{X}{\longrightarrow} A) \land (\phi \overset{Y}{\longrightarrow} A)
\Longrightarrow
\forall \textbf{a} \in A,
|Pr\big[\phi \overset{X}{\longrightarrow} \textbf{a}\big] - Pr\big[\phi \overset{Y}{\longrightarrow} \textbf{a}\big]|\;is\;negligible$
\end{center}

We will show that the game-based WHI definition (Section \ref{hi:weakhi}) is equivalent to
statement \ref{eq:whi}, that is, a data structure preserves WHI only if statement
\ref{eq:whi} is true. However, before we show the equivalence we point out the
necessity for the difference between conditions \ref{eq:hwhi} and \ref{eq:whi}.

As discussed in Section \ref{hitheory:preliminaries:hisemi}, there are two known ways to achieve history
independence. The first way is to make the ADT and data structure state
transition graphs isomorphic. The second way is to make the data structure
state transition graph randomized.
The requirement for identical memory distributions as per statement \ref{eq:hwhi}
rules out the use of randomization to achieve history independence\footnote{The use of randomization to achieve weak history independence is discussed in Section \ref{hi:randomization}.}. 
A randomized data structure implementation will rely on pseudo random generators.
The security of pseudo random generators relies on computational indistinguishability.
%Hence, for practicality, we relax the requirement from identical distributions
%to memory distributions that are computationally indistinguishable.
Therefore, the relaxed requirement of negligibility introduced in statement \ref{eq:whi}
is in fact not a limitation, but
rather a reconciliation of the definition by Hartline \etal with reality 
where we have computationally bounded adversaries.

Although Naor \etal proposed a WHI definition that requires identical distributions,
they also used randomization to design a history independent data structure.
%------------------------------------------------------------------------------------------------

\medskip
\noindent
\textbf{Equivalence of WHI definitions}
\smallskip

We now show that our gamed-based WHI definition (Section \ref{hi:weakhi}) is equivalent to a WHI
definition based on statement \ref{eq:whi}.

We rewrite statement \ref{eq:whi} for consistent notations as follows.
\begin{center}\label{eq:whi2}
$(s_{\phi} \overset{\delta_{0}}{\longrightarrow} s) \land (s_{\phi} \overset{\delta_{1}}{\longrightarrow} s)
\Longrightarrow
\forall s^{\mathcal{M}} \in s,
|Pr\big[s^{\mathcal{M}}_{\phi} \overset{\delta^{\mathcal{M}}_{0}}{\longrightarrow} s^{\mathcal{M}}\big] - Pr\big[s^{\mathcal{M}}_{\phi} \overset{\delta^{\mathcal{M}}_{1}}{\longrightarrow} s^{\mathcal{M}}\big]|\;is\;negligible
$\end{center}
Here, $\delta_{0}$ and $\delta_{1}$ are two ADT operation sequences that take
the ADT from initial state $s_{\phi}$ to state $s$.
$s_{\phi}$ and $s^{\mathcal{M}}_{\phi}$ are the initial ADT and the initial
data structure states, respectively.
$\delta^{\mathcal{M}}_{0}$ and $\delta^{\mathcal{M}}_{1}$ are the machine
programs corresponding to ADT operation sequences $\delta_{0}$ and $\delta_{1}$,
respectively.

%Specifically we show that if statement \ref{eq:whi2} is false (or true), then an adversary wins (or looses)
%in the WHI game.
History independence only considers cases where the condition
$(s_{\phi} \overset{\delta_{0}}{\longrightarrow} s) \land (s_{\phi} \overset{\delta_{1}}{\longrightarrow} s)$
is true, that is, both sequences $\delta_{0}$ and $\delta_{1}$ take the ADT to the
same end state $s$. Otherwise, the ADT states themselves reveal history.

We therefore have two cases to consider

\medskip
\textbf{Case 1:}
The distributions are computationally distinguishable, that is, 
\begin{center}
$\exists s^{\mathcal{M}} \in s\;such\;that\;|Pr\big[s^{\mathcal{M}}_{\phi} \overset{\delta^{\mathcal{M}}_{0}}{\longrightarrow} s^{\mathcal{M}}\big] - Pr\big[s^{\mathcal{M}}_{\phi} \overset{\delta^{\mathcal{M}}_{1}}{\longrightarrow} s^{\mathcal{M}}\big]|\;is\;non-negligible$.
\end{center}

Now consider the following adversarial strategy. Given a data structure state
$s^{\mathcal{M}}$ in step 4 of the WHI game, the adversary outputs $c$ such
that $\delta^{\mathcal{M}}_{c}$ has a higher probability of producing $s^{\mathcal{M}}$.
For such an adversarial strategy $\big|Pr[c' = c] - \frac{1}{2}\big|$ is
non-negligible for some $s^{\mathcal{M}}$. Therefore, the data structure
implementation does not preserve WHI.

\medskip
\textbf{Case 2:}
The distributions are computationally indistinguishable, that is, 
\begin{center}
$\forall s^{\mathcal{M}} \in s,
|Pr\big[s^{\mathcal{M}}_{\phi} \overset{\delta^{\mathcal{M}}_{0}}{\longrightarrow} s^{\mathcal{M}}\big] - Pr\big[s^{\mathcal{M}}_{\phi} \overset{\delta^{\mathcal{M}}_{1}}{\longrightarrow} s^{\mathcal{M}}\big]|\;is\;negligible$
\end{center}

In this case, from a computationally bounded adversary's perspective, the
representation $s^{\mathcal{M}}$ received in step 4 of the WHI game
is equally likely to have been produced by either $\delta^{\mathcal{M}}_{0}$ or $\delta^{\mathcal{M}}_{1}$.
Hence, observation of a data structure state gives the adversary
a negligible advantage in guessing $c$. The data structure implementation
therefore preserves WHI.
%---------------------------------------------------------------------------------

\medskip
\noindent
\textbf{Equivalence of SHI definitions}
\smallskip

For strong history independence Hartline \etal proposed the following definition.
\begin{definition}
\label{def:hshi}
{\em Strong History Independence (SHI)} \\
A data structure implementation is strongly history independent if, for any two (possibly empty)
sequences of operations X and Y that take a data structure in state A to state B,
the distribution over representations of B after X is performed on a
representation \textbf{\textsl{a}} is identical to the distribution after Y is
performed on \textbf{\textsl{a}}.
That is:
\end{definition}
\begin{center}
$(A \overset{X}{\longrightarrow} B) \land (A \overset{Y}{\longrightarrow} B)
\Longrightarrow
\forall$ \textbf{a} $\in A, \forall$ \textbf{b} $\in B$,
\textbf{Pr\big[a $\overset{X}{\longrightarrow}$ b\big] = Pr\big[a $\overset{Y}{\longrightarrow}$ b\big]}
\end{center}
In the above definition, $A \overset{X}{\longrightarrow} B$ denotes that a operation
sequence $X$ when applied to state $A$, results in state $B$. The notation \textbf{a} $ \in A$ means than
\textbf{a} is a memory representation of state $A$. \textbf{Pr\big[a $\overset{X}{\longrightarrow}$ b\big]}
denotes the probability that a sequence $X$ when applied to memory representation \textbf{a}, results
in representation \textbf{b}.

Similar to the case for WHI, our game-based SHI definition (Section \ref{hi:stronghi}) differs from the above definition
only by relaxing the requirement for identical distributions. That is, for SHI,
we require the following:
\begin{center}
$(A \overset{X}{\longrightarrow} B) \land (A \overset{Y}{\longrightarrow} B)
\Longrightarrow
\forall$ \textbf{a} $\in A, \forall$ \textbf{b} $\in B$,
$|$\textbf{Pr\big[a $\overset{X}{\longrightarrow}$ b\big] - Pr\big[a $\overset{Y}{\longrightarrow}$ b\big]}$|$ is negligible
\end{center}

The equivalence of SHI definitions follows similarly to the case of WHI.

%---------------------------------------------------------------------------------

\medskip
\noindent
\textbf{Summary of Differences}
\smallskip

The main differences between our definitions and the the definitions by Hartline \etal
are the following
\begin{itemize}
\item	The definitions by Hartline \etal are imprecise about the concepts of data structures,
	states, and memory representations. We precisely formalize all of these concepts.
\item	Hartline \etal do not consider the case of computationally bounded adversaries.
	We permit computationally bounded adversaries and thus have the
	negligibility definition instead of equality for memory distributions.
\end{itemize}

%\clearpage
%------------------------------------------------------------------------------------
\subsection{Canonical Representations}
\label{hi:canonical}

\begin{table*}[t]
\vspace{0.2cm}

\caption{Identification of scenarios where canonical representations are necessary
for history independence. N/A = not applicable.\label{hitheory:table:hicanrep}}{%
\centering
\begin{tabular}{|c|c|c|c|c|}
\hline
{\bf Programs of $\mathcal{M}$} & {\bf Random bits} & {\bf Adversary}       & {\bf History}	 &	{\bf Canonical}\\
				& {\bf hidden from} & {\bf computationally} & {\bf Independence} &	{\bf representations}\\
				& {\bf adversary}   & {\bf bounded}         & {\bf desired} 	 & 	{\bf needed?}\\
\hline
Randomized & Yes & Yes & WHI & No\\
\hline
N/A & N/A & No & N/A & Yes\\
\hline
N/A & N/A & N/A & SHI & Yes\\
\hline
Deterministic& N/A & N/A & N/A & Yes\\
\hline
\end{tabular}}
%\vspace{-0.2cm}
%\normalsize                  

\end{table*}

Canonically (or uniquely) represented data structures have the property
that each ADT state has a unique memory representation.
%represented by a unique machine (data structure) state,
Unique representation implies that the ADT and data structure
state transition graphs are isomorphic\footnote{Isomorphism is discussed in Section \ref{hitheory:preliminaries:hisemi}.}.
Canonically represented data structures give very strong guarantees for history independence and in many
cases are the only way to achieve history independence. 
%Hence, we discuss canonical representations in detail here.

We first define canonically represented data structures and then 
discuss several important results pertaining to canonical
representations and history independence.

We also summarize (Table \ref{hitheory:table:hicanrep}) the scenarios where canonical
representations are necessary for history independence across all combinations
of types of programs, secrecy of random bits, adversarial computational ability,
and the desired notion of history independence.

\begin{definition}
\label{def:canrep}
{\em Canonically represented data structure} \\
A data structure $\mathcal{D}$ implementing an ADT $\mathcal{A}$ on a
bounded RAM machine model $\mathcal{M}$ is canonically represented if
each ADT state has a unique memory representation, that is,
the mapping $m : \mathcal{S} \rightarrow 2^{\mathcal{S}^\mathcal{D}}$ is
injective and $|m(s)| = 1$, where $\mathcal{S}$ is the set of all ADT states,
$\mathcal{S}^\mathcal{D}$ is the set of all data structure states,
and $m(s)$ denotes the set of memory representations
of an ADT state $s \in S$ as per definition \ref{def:mr}.
%
%Let $\mathcal{A} = (\mathcal{S}, s_{\phi}, \mathcal{O}, \Gamma, \Psi)$ be an ADT,
%$\mathcal{M} = (\mathcal{S}^{\mathcal{M}}, s_{\phi}^{\mathcal{M}},
%\mathcal{P}^{\mathcal{M}}, \Gamma^{\mathcal{M}}, \Psi^{\mathcal{M}})$ be a bounded
%RAM machine model, and $\mathcal{D} = (s_{0}^{\mathcal{M}}, \alpha, \beta, \gamma)$
%be an implementation of $\mathcal{A}$ in $\mathcal{M}$, as per definitions
%\ref{def:adt}, \ref{def:rammodel} and \ref{def:ds} respectively.
%Let $m(s)$ be the memory representations of the ADT state $s$ as per definition
%\ref{def:mr}. Then, $\mathcal{D}$ is a canonically represented data structure
%if $m$ is a bijective mapping $m : \mathcal{S} \rightarrow \mathcal{S}^\mathcal{M}$.
\end{definition}
%----------------------------------------------------------------------------------
\subsubsection{ADTs with infinite states}
\label{hi:canonical:infinite}

The case of infinite ADT states is of particular importance for canonically represented
data structure implementations on a bounded RAM machine model.
Since the bounded RAM machine model has a finite number of available
bits, the machine state space is not large enough to provide a
unique representation for each ADT state when the ADT state space
is infinite. Impossibility of unique representations clearly suggests that
canonical representations for infinite state set ADTs are not possible
in practice since machines with infinite state space do not exists in reality. 
This straight-forwardly leads to the following theorem.
\begin{thm}
\label{hi:thm:canonimpossible}
Canonically represented data structure implementations for ADTs with infinite states
are impossible in practice.
\end{thm}

However, prior work \cite{naorantipersistence,naorcuckoo,golovinthesis} has claimed
designs for canonically represented data structures for the RAM model in
direct contradiction to Theorem \ref{hi:thm:canonimpossible}. 
%This contradiction is resolved by the following subtle fact. 
The contradiction arises from the fact that prior work has
implicitly considered ADTs with finite state space. Specifically,
the ADTs considered have have fewer states than the 
the total number of machine states.
%, that is ADTs with less than the possible number of machine states).
%Hence, from this point onwards we too only consider canonically represented data structure
%implementations for finite state space ADTs.

%For such data structures the mapping from ADT states
%to machine states is defined to be bijective (Section \ref{hi:canonical}).
%That is, each ADT state has a unique machine state representation.
%Now, since all possible ADT states may not be representable
%using a finite number of available machine bits, such a bijective mapping
%may not be possible. Which seems to suggest that uniquely represented data
%structures may not be possible in practice.

%However, canonically represented data structures have been proposed in
%prior work, which seems contradictory to our discussion above. 
%This contradiction is resolved by the following subtle point --
%In order to preserve history independence (in practice) we only
%require each ADT state, that can be represented using the finite bits
%of the machine to have a unique machine representation.
%Any non-representable ADT states may not be unique
%Hence, when referring to canonical data structures we are
%only concerned about the ADT states, the information in which
%can be represented using the finite number of machine bits. 

%For all such 
%ADT states the data structure is said to be uniquely represented.
%However, to achieve history independence it
%suffices for the ADT states to machine states mapping to be surjective.
%---------------------------------------------------------------------------------------
\subsubsection{Cannonical Representations and SHI}
\label{hitheory:canonical}
Since history independence was first proposed \cite{naorantipersistence},
it has been known that canonically represented data structures
support SHI. An interesting question posed in this context was whether
canonical representations are necessary to achieve SHI.
The question about the necessity of canonical representations for SHI was answered
by Hartline \etal
Hartline \etal \cite{hartlinecharacterizinghi} showed that SHI
cannot be achieved without canonical representations.

\begin{thm}
A data structure is strongly history independent iff it is
canonically represented.
\end{thm}

\begin{proof}
The proof by Hartline et al. \cite{hartlinecharacterizinghi} builds on the case that if a data structure is not
canonically represented, then an adversary can distinguish an empty sequence of operations
from a nonempty sequence of operations. In the context of our game based definition for SHI, 
we provide an equivalent proof for the same.

Consider an ADT $\mathcal{A}$ and a data structure $\mathcal{D}$ implementing $\mathcal{A}$ on a
bounded RAM machine model $\mathcal{M}$. Also assume that $\mathcal{D}$ is not canonically represented.
Now, Let $Q$ be an adversery and $C$ be the challenger in our game. $Q$ selects the following 
\begin{itemize}
\item Two ADT state $s_1$ and $s_2$.
\item Two sequence of operations $\delta_0$ and $\delta_1$; and two sequences of ADT inputs $I_0$ and $I_1$,
such that $\mathbb{O}(\delta_0,s_1,I_0) \rightarrow (s_2,\tau)$ and $\mathbb{O}(\delta_1,s_1,I_1) \rightarrow (s_2,\tau)$.
\end{itemize}

Let $\alpha_{1}$ and $\alpha_{2}$ be two distinct memory representations for ADT state $s_1$. 
We show that, with this setup, the adverary $Q$ can distinguish between an empty sequence of operations and 
a non empty sequence of operations. Consider that $Q$ selects $\delta_0$ to be an empty sequence of operation and $\delta_1$ 
to be a non empty sequence of operations. 
%Since an empty sequence of operations does not cause state changes 
%and $\delta_0$ being 	an empty sequence of operations implies $s_1 = s_2$. \\
In step 2 of the SHI game, the adversary sends $s_1$,$\delta_0$,$\delta_1$,$I_0$ and $I_1$ to $C$.
In step 3 of the SHI game, $C$ flips a coin and applies either $\delta_0$ or $\delta_1$ to $s_1$ and 
returns the memory representation of the output state to $Q$.
There are two possible cases for step 3 -
\begin{enumerate}
\item $C$ selects $\delta_0$ -- there is no change in the ADT state and the corresponding memory
representation since an empty sequence of operations does not cause state changes.
Hence in step 4, $C$ returns $\alpha_1$ to $Q$.
\item $C$ selects $\delta_1$ -- the final ADT state reached after performing all the operations in $\delta_1$ is $s_1$ but the memory representation for $s_1$ 
in this case may be either $\alpha_1$ or $\alpha_2$. In step 4, if $C$ returns $\alpha_2$ to $Q$, then $Q$ can correctly predict with 
non-negligible probability that $C$ has applied $\delta_1$ on $s_1$ to reach $\alpha_2$.  This breaks strong history independence 
for $\mathcal{D}$.
\end{enumerate}
\end{proof}

%For our SHI game definition from Section \ref{hi:stronghi} this means that
%if the adversary selects an empty sequence and a nonempty sequence of ADT
%operations in step 1 and if the data structure does not have canonical
%representations then the adversary's advantage is non-negligible.
%----------------------------------------------------------------------------------------
%Since achieving canonical representations is hard in practice we need to clearly
%identify the adversary models where canonical representations is a must.
%---------------------------------------------------------------------------------------

\noindent
\textbf{Canonical representations are not necessary for WHI}\\
\noindent
In the absence of canonical representations, it has been shown that
an adversary can distinguish an empty sequence of operations from
a nonempty sequence of operations thereby breaking SHI \cite{hartlinecharacterizinghi}.
If operation sequences are always assumed to be nonempty, canonical
representations are not necessary \cite{hartlinecharacterizinghi}.
We will define such a slightly relaxed notion of history independence
that permits only nonempty sequences in Section \ref{hi:deltahi}.
We show that WHI is preserved even for empty operation sequences
in the absence of canonical representations.

Consider the WHI game from Section \ref{hi:weakhi}.
The case in which the adversary selects two empty ADT operation sequences
in step 1 is trivial since empty sequences cause no state transitions and
hence there is no history to be revealed.

Now, consider the case when the adversary selects an
empty sequence $\delta_{\phi}$ and a nonempty sequence $\delta_{1}$ of
ADT operations. Both $\delta_{\phi}$ and $\delta_{1}$ are required
to take the ADT from the initial state to the same end state.
Since the empty sequence $\delta_{\phi}$ causes no state transitions, 
end state for both sequences $\delta_{\phi}$ and $\delta_{1}$  will be the initial
ADT state.

Then, in step 3, the challenger chooses either $\delta_{\phi}$ or $\delta_{1}$ 
and sends the resulting memory representation to the adversary.
Since the end state for the two operation sequences is the initial ADT state,
the memory representation sent to the adversary in step 4 will be
the data structure initialization state. From the data structure definition
(Section \ref{hitheory:preliminaries:ds}), we know that the initial ADT state
has a corresponding fixed unique memory representation. Hence, irrespective of the
nonempty sequence that the adversary selects in step 1, the adversary receives
the initial ADT state's memory representation in step 4.
Since the adversary receives the same representation each time, the
adversary gains no advantage in guessing
whether $\delta_{\phi}$ or $\delta_{1}$ was chosen by the challenger in step 3.

ADT states other than the initial ADT state can have multiple memory representations.
Multiple representations for ADT states does not break WHI as long it is ensured that 
from the adversary's perspective, all representations of the current ADT state
are equally likely to be observed. Randomization achieves equal 
likelyhood for all representations of an ADT state (Section \ref{hi:randomization}).
In Section \ref{hitheory:preliminaries:hisemi} we covered the use of randomization for WHI from
the perspective of state transition graphs. 
Later, in Section \ref{hi:randomization} we will show how to achieve WHI by randomization in practice.
, which can be used for WHI.
For deterministic machine programs (also
described in Section \ref{hi:randomization}) WHI too requires canonical representations.
%---------------------------------------------------------------------------------------

\subsubsection{Canonical representations and adversary models}
Canonically represented data structures are history independent in the strongest
sense, secure even against a computationally unbounded adversary \cite{golovinthesis}.
For a computationally unbounded adversary, canonical representations are also necessary
for WHI.

%------------------------------------------------------------------------------------
\subsection{Randomization and HI}
\label{hi:randomization}

In Section \ref{hitheory:preliminaries:hisemi}, we introduced the use of 
randomization for WHI from the point of view
of state transition graphs. 
%We showed that 
%randomization involves ensuring that for any two ADT states $s_{0}$ and $s_{1}$,
%if there is a path in the ADT state transition graph between $s_{0}$ and $s_{1}$,
%then there must be a path from all memory representations
%of $s_{0}$ to all memory representations of $s_{1}$ in the
%data structure's state transition graph. The choice of path in the data
%structure state transition graph between representations of $s_{0}$ and
%$s_{1}$ is then made at random.

In practice, randomization is achieved using the machine
programs implementing the ADT operations. 
An ADT operation $o$ takes the ADT from a state $s_{1}$ to a state $s_{2}$.
%(see ADT definition in Section \ref{hitheory:preliminaries:adt}).
A machine program implementing $o$ takes the data
structure from a memory representation of state $s_{1}$ to a memory
representation of state $s_{2}$.
%The states here are the memory representations of the corresponding ADT states.
Since each ADT state can have several memory representations (Section \ref{hitheory:preliminaries:hisemi}),
the program has a choice amongst all representations of state $s_{2}$
and picks one representation as the result of a transition.
Starting from a fixed memory representation of $s_{1}$, and a fixed input, if the
program takes the data structure to a fixed
resulting representation of $s_{2}$ on each execution, then the program is said to be deterministic.
If on each execution the resulting representation is chosen uniformly at random 
from all possible representations of state $s_{2}$,
then the program is said to be randomized.

To illustrate, consider an ADT operation $o$ and a machine
program $p$ implementing $o$. Let $o(s_{1}, i) \rightarrow (s_{2}, \tau)$ 
denote the transition from ADT state $s_{1}$ to ADT state $s_{2}$
using an ADT input $i$ and producing an ADT output $\tau$.
Also,
%let $s^{\mathcal{M}}_{1}$ be a memory representation of state $s_{1}$,
let $m(s_{1})$ and $m(s_{2})$ denote the set of memory representations of states
$s_{1}$ and $s_{2}$, respectively.
%$s^{\mathcal{M}}_{2}, ..., s^{\mathcal{M}}_{m}$
%be all possible memory representations of ADT state $s$.
Then, for history independence, the following must hold for program $p$:
% must be such that
\begin{center}
$Pr[p(s^{\mathcal{M}}_{1}, \alpha(i)) \rightarrow (s^{\mathcal{M}}_{2}, \beta(\tau))] = \frac{1}{|m(s_{2})|}$,
$\forall$ $s^{\mathcal{M}}_{1} \in m(s_{1})$ and $\forall$ $s^{\mathcal{M}}_{2} \in m(s_{2})$.
\end{center}
Here,
%$Pr[E]$ denotes the probability of event $E$, and
$\alpha(i)$ and $\beta(\tau)$ are
the machine representations of ADT input $i$ and ADT output $\tau$, respectively.
%and $m(s_{2})$ denotes the set of all memory representations of state $s_{2}$
%(refer to Section \ref{hitheory:preliminaries:memrep} for definition of memory representations).

Randomization here refers to the selection of memory representations
for ADT states and not to program outputs. A program's output is the machine
representation of the corresponding ADT operation's output.

If randomization is used for history independence, then
%the random bits of the machine 
random choices made by the machine programs must be hidden from
the adversary. If the adversary has knowledge of the random bits, then
from the adversary's point of view the machine programs are deterministic.
Data structures with deterministic machine programs require canonical representations.
%(implementing the ADT operations)

\section{Generalizing History Independence}
\label{hitheory:hi:generalizinghi}

SHI is a very strong notion of history independence
requiring canonical representations \cite{golovinthesis,hartlinecharacterizinghi}.
%If a data structure is strongly history independent it hides all information about its history
%other than what is inherently visible from the ADT states.
Canonically represented data structures are not efficient \cite{buchbinderbounds}.
For heap and queue data structures Buchbinder \etal \cite{buchbinderbounds} show that 
certain operations that require logarithmic time under WHI take linear
time under SHI.
%Our results (Section \ref{hifs:experiments} also show
%that achieving SHI reduces system performance.
Hence, it is worth to question the need for canonical representations
for history independence.
Many scenarios may not require such a strong notion 
%For example, applications where only partial history is required to be hidden,
rendering SHI data structures with canonical representations
inefficient.

Some scenarios that can be efficiently realized
by new history independence notions:
\begin{itemize}
\item	Hiding evidence of specific operations only. For example, hiding only the fact
	that a specific data item has been deleted in the past, as required by regulations \cite{cfr240}%,pipeda,eu-drd}. 
	%Eliminating evidence
	%of past deletes is directly applicable for regulatory compliance.
	%Retention Regulations \cite{cfr240,pipeda,eu-drd} are only concerned with hiding the
	%past existence of deleted data and not with other aspects of history, such as the insertion
	%sequence of current data.
\item	A most recently used (MRU) caching or a journaling system by definition reveals the last $k$ operations.
	%Here, an adversary may be permitted to
	%determine the last $k$ operations performed \cite{naorantipersistence,hifs}.
	Hence, journaling and caching require a new notion of history independence, wherein
	no history is revealed other than the last $k$ operations \cite{naorantipersistence}.
\item	Revealing only the number of times each operation is performed \cite{golovinthesis}.
	For example, in a file-sharing application disclosing file-access counts may be permissible,
	but not the access order.
\end{itemize}

%Existing work has voiced the requirement to accommodate such relaxed notions of
%history independence 
%Existing work \cite{naorantipersistence,golovinthesis} has already suggested
%that for efficiency, it is important to benefit from new relaxed notions
%of history independence.
%A proper theoretical framework is needed to precisely
%define new history independence notions.
%In the following, we take first steps towards such a framework.

A straight-forward way to define new notions of history independence is to
provide a new game-based definition for each scenario. However,  
defining distinct scenario-specific games can quickly become a
tedious process.
Instead, we introduce a  definitional framework
%generic game-based definition of history independence
that can accommodate a broad spectrum of history independence notions.
We term the new framework as $\Delta$ history independence ($\Delta$HI),
where $\Delta$ is the parameter determining the history independence
flavor. As we shall see, $\Delta$HI also captures both WHI and SHI.
In addition, $\Delta$HI helps to reason about the history revealed or
concealed by existing data structures which were designed without
history independence in mind.
%-------------------------------------------------------------------------------------
\subsection{$\Delta$ History Independence ($\Delta$HI)}
\label{hi:deltahi}

The WHI and SHI games (Sections \ref{hi:weakhi} and \ref{hi:stronghi}
respectively) are defined over a subset of ADT operation sequences.
For WHI, the adversary is permitted to select sequences
that take the ADT from initialization to the same end state.
For SHI, the permitted sequences are ones that take the ADT from
the same starting state to the same ending state. The selection
is made by the adversary in step 1 of both the WHI and SHI games.
%Steps 2-5 are are similar in both games. 
Hence, the initial selection permitted to the adversary determines the
history that is desired to be revealed or hidden.
By generalizing the selection step, we can accommodate a broad spectrum
of history independence notions. 
%We do this in $\Delta$ history independence, which is
We achieve the generalization in $\Delta$HI, defined by the following game:
\begin{framed}
\begin{footnotesize}
Let $\mathcal{A} = (\mathcal{S}, s_{\phi}, \mathcal{O}, \Gamma, \Psi)$ be an ADT, $\mathcal{M} = (\mathcal{S}^{\mathcal{M}}, s_{\phi}^{\mathcal{M}}, \mathcal{P}^{\mathcal{M}}, \Gamma^{\mathcal{M}}, \Psi^{\mathcal{M}})$ be a bounded RAM machine model, 
and $\mathcal{D} = (\alpha, \beta, \gamma, s_{0}^{\mathcal{M}})$ be a data structure implementing $\mathcal{A}$ in $\mathcal{M}$, as per definitions \ref{def:adt}, \ref{def:rammodel} and \ref{def:ds}, respectively.
Also, let $\zeta$ be the set of all ADT operation sequences, $\Upsilon$ be the set of all ADT input sequences, and
%from $\mathcal{O}$, 
$\Delta$ be a function $\Delta: \mathcal{S} \times \mathcal{S} \times \zeta \times \zeta \times \Upsilon \times \Upsilon \rightarrow \{0, 1\}$.

\begin{enumerate}
\item	A probabilistic polynomial time-bounded adversary selects the following.
	\begin{itemize}
	\item	Two ADT states $s_{1}$ and $s_{2}$;
		two sequences of ADT operations $\delta_{0}$ and $\delta_{1}$; and two sequences of ADT inputs $I_{0}$ and $I_{1}$;
		such that
				%$\mathbb{O}(\delta_{0}, s_{1}, I_{0}) \rightarrow (s_{2}, \tau)$ and
				%$\mathbb{O}(\delta_{1}, s_{1}, I_{1}) \rightarrow (s_{2}, \tau)$, and
				$\Delta(s_{1}, s_{2}, \delta_{0}, \delta_{1}, I_{0}, I_{1}) = 1$.
				%Both $\delta_{1}$ and $\delta_{2}$ take the ADT from state $s_{1}$
				%to state $s_{2}$ producing the same output $\tau$.
	\item 	A memory representation $s_{1}^{\mathcal{M}}$ of ADT state $s_{1}$.
	\end{itemize}
\item	The adversary sends $s_{1}$, $s_{1}^{\mathcal{M}}$, $\delta_{0}$, $\delta_{1}$, $I_{0}$ and $I_{1}$ to the challenger.
\item	The challenger flips a fair coin $c \in \{0,1\}$ and computes
	$\mathbb{O}^{\mathcal{M}}(\delta_{c}^{\mathcal{M}}, s^{\mathcal{M}}_{1}, I_{c}) \rightarrow (s^{\mathcal{M}}, \tau^{\mathcal{M}})$,
	where $\delta_{c}^{\mathcal{M}} = \chi(\delta_{c})$.
	%and $\tau^{\mathcal{M}} = \beta(\tau)$.
	That is, the challenger applies the program sequence $\delta_{c}^{\mathcal{M}}$
	corresponding to the ADT operation sequence $\delta_{c}$ to the data
	structure state $s^{\mathcal{M}}_{1}$, resulting in a memory representation
	$s^{\mathcal{M}}$, 
	%of state $s_{2}$,
	and a machine output $\tau^{\mathcal{M}}$.
\item	The challenger sends the memory representation $s^{\mathcal{M}}$ and the machine output $\tau^{\mathcal{M}}$ to the adversary.
\item	The adversary outputs $c' \in \{0,1\}$.
\end{enumerate}

\begin{center}
The adversary wins the game if $c' = c$.
\end{center}

$\mathcal{D}$ is said to be $\delta$ history independent if the advantage of the adversary for winning the game, 
defined as $|Pr[c^{'}=c]-1/2|$ is negligible
(where ``negligible'' is defined over any implementation-specific security parameters of the programs in $\mathcal{P^M}$).
\end{footnotesize}
\end{framed}

%\subsubsection{The Role of $\Delta$}
%\label{hi:deltahi:deltarole}

Function $\Delta$ determines the pairs of ADT states, ADT operation sequences, and
ADT input sequences that the adversary is permitted to select in step 1
of the $\Delta$HI game. For the adversary-selected ADT states, operation sequences, and input sequences,
the $\Delta$HI game can be played and the data structure implementation is required to ensure that
the advantage of the adversary is negligible.
Thus, for a given ADT, $\Delta$ defines two sets,
\begin{center}
$H_{\Delta} = \{ (s_{1}, s_{2}, \delta_{0}, \delta_{1}, I_{0}, I_{1})\;|\; \Delta(s_{1}, s_{2}, \delta_{1}, \delta_{2}, I_{0}, I_{1}) = 1\}$, and
$\overline{H}_{\Delta} = \{ (s_{1}, s_{2}, \delta_{0}, \delta_{1}, I_{0}, I_{1})\;|\; \Delta(s_{1}, s_{2}, \delta_{1}, \delta_{2}, I_{0}, I_{1}) = 0\}$.
\end{center}
For all tuples in $H_{\Delta}$, history independence is preserved, that is, neither the ADT nor the
data structure implementation reveals the operation sequence 
%$\delta_{0}$ or $\delta_{1}$
selected by the challenger in step 3.
For all tuples in $\overline{H}_{\Delta}$, history independence is not required to be preserved
since the ADT itself reveals the sequence of operations used.

A careful choice of $\Delta$ allows us to precisely define both SHI and WHI,
and a broad spectrum of new history independence notions.
%--------------------------------------------------------------------
%\subsection{Relaxing History Independence}
%\label{hi:deltahi:examples}
%
%In order to define any notion of history independence including WHI and SHI we 
%can use the $\Delta$ history independent game as is. The additional
%task is to precisely define the function $\Delta$. 
%Before we illustrate this for several new history independence notions,
%we note the following important point.
%SHI is the most stringent notion of history
%independence \cite{golovinthesis}. Hence, any other notion of history
%independence will necessarily be less stringent in comparison, and doing
%so can be viewed as relaxing history independence.
In the following, we illustrate the use of $\Delta$HI framework
to define some familiar history independence notions and a few
previously unconsidered notions of history independence.
%------------------------------------------------
%\medskip
%\noindent
%{\bf Strong History Independence (SHI)~}
\subsubsection{Strong History Independence (SHI)}
\label{hi:deltahi:shi}

We discussed SHI in Section \ref{hi:stronghi}.
Here, we define the function $\Delta$ for SHI.
\begin{small}
\begin{center}
$\Delta(s_{1}, s_{2}, \delta_{0}, \delta_{1}, I_{0}, I_{1}) =
	\left\{ \begin{array}{cl}
			1 & \mbox{if}\; \mathbb{O}(\delta_{0}, s_{1}, I_{0}) \rightarrow (s_{2}, \tau)\;and\;
			  		\mathbb{O}(\delta_{1}, s_{1}, I_{1}) \rightarrow (s_{2}, \tau)\\
			0 & \mbox{otherwise}
\end{array}\right.$
\end{center}
\end{small}

For SHI, the adversary's advantage in the $\Delta$HI game must be negligible
when in step 1, the adversary selects any two ADT operation sequences that take
the ADT from a state $s_{1}$ to a state $s_{2}$ producing the same ADT output $\tau$.
%------------------------------------------------
%\medskip
%\noindent
%{\bf Weak History Independence (WHI)~}
\subsubsection{Weak History Independence (WHI)}
\label{hi:deltahi:whi}

Refer to Section \ref{hi:stronghi} for discussion on WHI, which requires the following
definition of $\Delta$.
\begin{small}
\begin{center}
$\Delta(s_{1}, s_{2}, \delta_{0}, \delta_{1}, I_{0}, I_{1}) =
	\left\{ \begin{array}{cl}
			1 & \mbox{if}\; s_{1} = s_{\phi}\;and\;
			  		\mathbb{O}(\delta_{0}, s_{1}, I_{0}) \rightarrow (s_{2}, \tau)\;and\;\\
			  &  		\;\;\;\mathbb{O}(\delta_{1}, s_{1}, I_{1}) \rightarrow (s_{2}, \tau)\\
			0 & \mbox{otherwise}
\end{array}\right.$
\end{center}
\end{small}

Since WHI permits the adversary to observe a single data structure state,
the adversary chooses only the end
state $s_{2}$ in step 1 of the $\Delta$HI game.
The starting state on which sequences $\delta_{0}$
and $\delta_{1}$ are applied is the initial ADT state $s_{\phi}$.
%------------------------------------------------
%$ $\\
%\noindent
%{\bf Append Only Log (No history independence)~}

\subsubsection{Null history independence ($\phi$HI)}
\label{hi:deltahi:appendonlylog}

Under null history independence, a data structure
conceals no history except for the trivial case when the ADT operation
sequences and ADT input sequences selected by the adversary in the
$\Delta$HI game are identical.
Example of a data structure with $\phi$HI is an append-only log.
% protects no historical information.
We can reflect $\phi$HI using the following.

\begin{small}
\begin{center}
$\Delta(s_{1}, s_{2}, \delta_{0}, \delta_{1}, I_{0}, I_{1}) =
	\left\{ \begin{array}{cl}
			1 & \mbox{if}\; \mathbb{O}(\delta_{0}, s_{1}, I_{0}) \rightarrow (s_{2}, \tau)\;and\; \\
			  & 		\mathbb{O}(\delta_{1}, s_{1}, I_{1}) \rightarrow (s_{2}, \tau)\\
			  &		\;\;\;and\;\delta_{1} = \delta_{2}\;and\;I_{1} = I_{2}\\
			0 & \mbox{otherwise}
\end{array}\right.$
\end{center}
\end{small}

%------------------------------------------------
%$ $\\
%\noindent
%{\bf SHI*~}

\subsubsection{SHI*}
\label{hi:deltahi:shi2}

%In Section \ref{hi:canonical}, we discussed the necessity
%of canonical representations for SHI.
The necessity of canonical representations for SHI was
proven by Hartline \etal \cite{hartlinecharacterizinghi}.
The proof by Hartline \etal \cite{hartlinecharacterizinghi} builds on the case
that if a data structure is not canonically represented, then an
adversary can distinguish an empty sequence of operations from a
nonempty sequence.
Hartline \etal \cite{hartlinecharacterizinghi} then proposed SHI*, which
is defined over nonempty ADT operation sequences.
SHI* data structures were initially expected to more efficient
than data structures providing SHI.
However, Hartline \etal \cite{hartlinecharacterizinghi} found that SHI* still
poses very strict requirements on a data structure and may not
differ from SHI in asymptotic complexity. 

Here, we give the $\Delta$ function for SHI*.

\begin{small}
\begin{center}
$\Delta(s_{1}, s_{2}, \delta_{0}, \delta_{1}, I_{0}, I_{1}) =
	\left\{ \begin{array}{cl}
			1 & \mbox{if}\; \mathbb{O}(\delta_{0}, s_{1}, I_{0}) \rightarrow (s_{2}, \tau)\;and\;\\
			  & 		\mathbb{O}(\delta_{1}, s_{1}, I_{1}) \rightarrow (s_{2}, \tau)\;and\;\\
			  & 		\;\;\;|\delta_{0}| > 0\;and\;|\delta_{1}| > 0\\
			0 & \mbox{otherwise}
\end{array}\right.$
\end{center}
\end{small}

SHI* closely resembles SHI except that the operations sequences $\delta_{0}$
and $\delta_{1}$ must be nonempty.

%------------------------------------------------
%$ $\\
%\noindent
%{\bf Reveal last k operations (MRU Cache, File System Journal)~}
\subsubsection{Reveal last k operations (Most Recently Used Cache, Journal)}
\label{hi:deltahi:lastk}

System features such as caching and journaling by
definition reveal the last k operations performed from the
ADT state itself. Thus, for caching and journaling, we need to define a
$\Delta$ function, such that no additional historical
information is leaked from the memory representations
other than the last k operations. We define the new notion as follows. 

Let $\delta[i]$ denote the $i^{th}$ operation in the sequence $\delta$.
Also, let $\delta[i, j]$ denote a subsequence of $\delta$, \{$\delta[i], \ldots, \delta[j]$\} with $i \leq j$.
%That is, $\delta[i, 
\begin{small}
\begin{center}
$\Delta(s_{1}, s_{2}, \delta_{0}, \delta_{1}, I_{0}, I_{1}) =
	\left\{ \begin{array}{cl}
			1 & \mbox{if}\; \mathbb{O}(\delta_{0}, s_{1}, I_{0}) \rightarrow (s_{2}, \tau)\;and\;\\
			  &		\mathbb{O}(\delta_{1}, s_{1}, I_{1}) \rightarrow (s_{2}, \tau)\;and\;\\
			  & 		\;\;\;|\delta_{0}| \geq k\;and\;|\delta_{1}| \geq k\;and\;\\
			  &		\delta_{0}[|\delta_{0}|-k,|\delta_{0}|] = \delta_{1}[|\delta_{1}|-k,|\delta_{1}|]\\
			0 & \mbox{otherwise}
\end{array}\right.$
\end{center}
\end{small}

The adversary is permitted to choose
%(in step 1 of $\Delta$HI game)
two sequences $\delta_{0}$ and $\delta_{1}$, such the that last $k$ operations
in $\delta_{0}$ and $\delta_{1}$ are the same. Other than the last
$k$ operations, sequences $\delta_{0}$ and $\delta_{1}$ may differ.
Yet, the adversary should be unable to identify the sequence chosen
in step 3.
%------------------------------------------------
\medskip
\noindent

\subsubsection{Operation-Agnostic History Independence (OAHI)}
\label{hi:deltahi:oahi}
Consider a secure deletion application that wishes to
destroy any evidence of a delete operation performed
in the past. That is, an adversary
(by observing the memory representations)
should be unable to detect whether a delete operation was
performed or not other than guessing.
In general, any particular operation may require to be concealed, not just
deletes.
can be extended to any ADT operation (not just delete)
We introduce a new notion of history independence that conceals
specific ADT operations. The new notion is referred to as 
operation-agnostic history independence (OAHI).  
A data structure that is $\Delta$ history independent
given the following $\Delta$ function
guarantees operation-agnostic history independence
for an ADT operation $o$.

\begin{center}
\begin{small}
$\Delta(s_{1}, s_{2}, \delta_{0}, \delta_{1}, I_{0}, I_{1}) =
	\left\{ \begin{array}{cl}
			1 & \mbox{if}\; \mathbb{O}(\delta_{0}, s_{1}, I_{0}) \rightarrow (s_{2}, \tau)\;and\;
			  		\mathbb{O}(\delta_{1}, s_{1}, I_{1}) \rightarrow (s_{2}, \tau)\;and\;\\
			  &		\;\;\;o \in \delta_{0}\;and\;o \notin \delta_{1}\\
			0 & \mbox{otherwise}
\end{array}\right.$
\end{small}
\end{center}

In OAHI, neither the presence of operation $o$ in $\delta_{0}$, nor the absence
of $o$ in $\delta_{1}$ gives the adversary any advantage in guessing
the sequence chosen by the challenger in step 3 of the $\Delta$HI game.

\subsection{Measuring History Independence}
\label{hi:deltahi:measuring}

We have seen that new notions of history independence can be easily
derived from $\Delta$ history independence by defining the
appropriate $\Delta$ function. In this section, we present an  
intuitive way of comparing $\Delta$ functions on
the basis of the history they require to be concealed or preserved.

%Let, $H_{\Delta}$ be a subset of $\mathcal{S} \times \mathcal{S} \times \zeta \times \zeta$
%such that $\Delta$ outputs 1 for all quadruples in $H_{\Delta}$.
%Let, $H_{\Delta} = \{ (s_{1}, s_{2}, \delta_{1}, \delta_{2})\;|\; \Delta(s_{1}, s_{2}, \delta_{0}, \delta_{1}, I_{0}, I_{1}) = 1\}$,
%where $s_{1}$ and $s_{2}$ are any two ADT states, and $\delta_{1}$ and $\delta_{2}$
%are any two ADT operation sequences.
For a given $\Delta$ function we defined the set $H_{\Delta}$ (Section \ref{hi:deltahi})
that represents all combinations
of ADT states, operation sequences, and ADT input sequences for
which the adversary's advantage is negligible in the $\Delta$ history
independence game. That is, for all members of $H_{\Delta}$,
history independence is preserved. One insight is
to use the cardinality of $H_{\Delta}$ as a measure of history independence.

Recall from Section \ref{hitheory:preliminaries:ds} that an ADT
can have several data structure implementations.
Let $\mathcal{D}$ and $\mathcal{D'}$ be two implementations of
an ADT $\mathcal{A}$, such that $\mathcal{D}$ is $\Delta$ history
independent and $\mathcal{D'}$ is $\Delta'$ history
independent for two functions $\Delta$ and $\Delta'$.
Now, we say that $\mathcal{D}$ is more history independent
than $\mathcal{D'}$ if $H_{\Delta'} \subset H_{\Delta}$.

Note that $|H_{\Delta}| > |H_{\Delta'}|$ alone does
not imply that $\mathcal{D}$ is more history independent
than $\mathcal{D'}$ since an application may be
more sensitive to the history preserved by $D'$ than the history
preserved by $D$.
Only in the case where $H_{\Delta'} \subset H_{\Delta}$ can
we consider $\mathcal{D}$ to be a more history independent implementation
than $\mathcal{D'}$.
%based on the cardinalities of $H_{\Delta}$ and $H_{\Delta'}$ alone.
%In all other cases the specific application history-sensitive scenarios
%also need to be taken into account.

%----------------------------------------------------------------------------------------------
\subsection{Deriving History Independence}
\label{hi:deltahi:deriving}

In order to provide a history independent implementation
for an ADT, we first require the $\Delta$ function to
be precisely defined. Then, a history independent data structure
can be designed that satisfies the $\Delta$ function.
Satisfying a  $\Delta$ function means that the adversary's advantage 
is always negligible in the $\Delta$ history independence game.
In effect, so far we have approached history independence as a 
define-then-design process.

However, data structures have been in use for a long time and
most data structures have been designed for efficiency or functionality
with no history independence in mind. A natural question
then arises -- are there any meaningful\footnote{$\Delta = 0$
is satisfied by all data structures. Hence, we need
to determine $\Delta$ functions that are more useful in practice.}
$\Delta$ functions satisfied by existing data structures?.

%The problem can be precisely stated as follows:
%\emph{Given a data structure implementation $\mathcal{D}$ of an ADT $\mathcal{A}$
%on a machine model $\mathcal{M}$, determine a uncontained $\Delta$ function,
%such that $\mathcal{D}$ is $\Delta$ history independent.}

A data structure can be $\Delta$ history independent for several
$\Delta$ functions. For example, a data structure that satisfies SHI,
also satisfies WHI, OAHI, and OIAHI.
Hence, for a given data structure $\mathcal{D}$ finding a $\Delta$ function
may not be a particularly difficult task. It may be more useful instead
to determine an uncontained $\Delta$ function for $\mathcal{D}$.
We define an uncontained $\Delta$ function for a data structure
as follows.

\begin{definition}
\label{def:uncontaineddelta}
{\em Uncontained $\Delta$ function} \\
A $\Delta$ function for a data structure $\mathcal{D}$ is uncontained if
$\mathcal{D}$ is $\Delta$ history
independent and $\nexists\;\Delta'$, such that 
$\mathcal{D}$ is also $\Delta'$ history independent and $H_{\Delta} \subset H_{\Delta'}$,
where 

$H_{\Delta} = \{ (s_{1}, s_{2}, \delta_{0}, \delta_{1}, I_{0}, I_{1})\;|\\$ $\; \Delta(s_{1}, s_{2}, \delta_{0}, \delta_{1}, I_{0}, I_{1}) = 1\}$;\\
$H_{\Delta'} = \{ (s_{1}', s_{2}', \delta_{0}', \delta_{1}', I_{0}', I_{1}')\;|\; \Delta'(s_{1}', s_{2}', \delta_{0}', \delta_{1}', I_{0}', I_{1}') = 1\}$;
$s_{1}$, $s_{2}$, $s_{1}'$, and $s_{2}'$ are ADT states; $\delta_{0}$, $\delta_{1}$, $\delta_{0}'$, and $\delta_{1}'$ are ADT operation sequences; and $I_{0}$, $I_{1}$, $I_{0}'$, and $I_{1}'$ are ADT input sequences.

\end{definition}

We can determine an uncontained $\Delta$ function for existing data structures
on a case-by-case basis. An open question is whether there exists
a general mechanism for deriving an uncontained $\Delta$ function for a
given data structure.

\section{From Theory To Practice}
\label{hitheory:theory2practice}

\subsection{Defining Machine States}
\label{hitheory:theory2practice:components}

The RAM model of execution described in Section
\ref{hitheory:preliminaries:execmodels} consists of two
components, the RAM and the CPU. Hence, the machine
state for the RAM model includes bits from both
the RAM and the CPU. In general,
%since a system is composed of several components,
the machine state for a system-wide machine model
will comprise all system component states. 
A system-wide history independent implementation has to then consider
each individual component's characteristics along the interaction
between the components.
Providing system-wide history independence is therefore challenging.

However, in practice an adversary may have access to only a subset
of system components. In this case,
for the purpose of history independence, the machine
state can be defined over the adversary-accessible
components only. For example, history independent data structures 
proposed in existing work (Section \ref{hitheory:related})
are designed with the RAM model in mind. However, the machine states
considered for history independence only include
bits from the RAM and exclude the CPU.
%----------------------------------------------------------------------------------------------------------
\subsection{Building History Independent Systems}
\label{hitheory:theory2practice:buildinghisystems}

Various techniques for designing history independent data
structures for commonly used ADTs
such as queues, stacks, and hash tables have been proposed \cite{golovinthesis}.
Our focus on the other hand is designing \emph{systems} with
end-to-end history independent characteristics.
The difference between history independent implementations
for simple ADTs, such as stacks and queues versus a complete system,
such as a database, or a file system is a matter
of often exponentially increasing complexity. Fundamentally,
any system can be modeled as an ADT and an history independent implementation
can be sought for the system.

We introduce a general recipe for building history independent systems as follows:

\begin{enumerate}
\item	Model the system as an ADT. For a specific example of file
	system as an ADT, refer to section \ref{hi:hifs}.
\item	Select a machine model for implementation. While defining the
	machine state identify all machine components that the adversary
	has access to and define the machine state associated with the
	adversary-accessible components.
\item	Depending on the application scenario, fix a desired notion
	of history independence and the corresponding $\Delta$ function.
\item	Based on the definition of $\Delta$, provide an implementation
	over the selected machine model. For complex systems, the implementation
	will likely require the most effort since the machine programs
	implementing the ADT operations must provably ensure that
	the advantage of the adversary is negligible in the $\Delta$HI game.
\end{enumerate}

In section \ref{hi:hifs}, we follow the above recipe
to design a history independent file system.

\section{On A Philosophical Note}
\label{hitheory:philosophy}

At a very high level, the motivation for history independence can
be stated as follows.
\begin{center}
\emph{For any logical state $S_{L}$, the physical
state $S_{P}$ representing $S_{L}$ may reveal information 
about the history leading to $S_{L}$, that is otherwise not discernible
via solely $S_{L}$.}
\end{center}

So far, we have considered
the logical state to be the ADT state and the physical state to be the
underlying machine state representing the ADT state, that is, the physical
state is the set of all bits of the machine. Our selection
of logical and physical states seems rather arbitrary. We do this
specific selection due to our adversary model,
which assumes that the adversary can interpret information
at the level of bits. An adversary, that can for example, examine
the electric charge in individual capacitors used to represent
the bits will require a different choice of logical and physical state
descriptions. A straight-forward choice would be to consider a bit as
a logical state and the precise capacitor state as the physical state.  

The following interesting question arises from this discussion -- 
\emph{is history independence only a matter of perspective}?. The short
answer is \emph{yes}, history independence is a matter of perspective.
There is no universal history independence.

To clarify, consider the universe as a whole from the viewpoint of classical
physics. Under the classical viewpoint, knowledge of current
state of all objects in the universe enables determination of
any past or future universal state since the laws of physics work
both forwards and backwards in time. Hence, the past is never hidden
and history independence is impossible. For example, using the
currently observed movement of galaxies, the past states of the
universe can be inferred up to the very initial moments of the big bang.

Physical phenomena at the subatomic scale is explained by quantum
physics.
At the quantum level, the universe appears nondeterministic.
Further, the
uncertainty principle \cite{rae2005quantum} restricts the ability to accurately
measure the current state of a quantum system. Since the current
state cannot be accurately known, it may seem the past states
cannot be determined either and history independence can be
achieved at the quantum level.

However, even at the quantum level history
independence is still a matter of perspective.
The perspective is governed by the interpretation of quantum
physics used. Under the many-worlds interpretation,
the multiverse as a whole is deterministic \cite{everetttheory}.
The probabilistic nature at the quantum level is only
our perception since our observations are limited to a single
universe.
A hypothetical all-powerful adversary that can view the entire
multiverse would have a full view of the past and the future similar
to the case of classical physics making history independence
in the presence of such an adversary impossible.
%We note that current technology limits the processing of
%information to 
%As far as we know, the basic unit of matter are the subatomic
%particles. Thus, the physical state at the lowest level for any
%logical state is the state of the set of subatomic
%particles of the physical system.
%The finest granularity of information an adversary
%can therefore discern is the state of all subatomic particles
%representing the logical ADT states (ignoring the uncertainly
%principle).
%Hence, if all layers up to the physical subatomic level exhibit
%history independence (e.g., use canonical representations), true
%history independence can be achieved. We note that this is far beyond
%current technology, but is interesting from a philosophical perspective.

%
\section{Practical SHI for File Systems}
\label{hi:hifs}
%

%In previous sections we layed the theoretical foundations for history independence.
%We explored the concepts of ADTs, machine models, data structures, 
%and memory representations. We then formalized history independence and
%introduced the $\Delta$ history independence framework.

We now apply our theoretical concepts
and results towards practical history independent system designs.
% (Sections \ref{hitheory:preliminaries} - \ref{hitheory:theory2practice}),

%Various techniques for designing history independent data
%structures for commonly used ADTs
%such as queues, stacks, and hash tables have been proposed \cite{golovinthesis}.
Our focus %on the other hand 
is designing \emph{systems} with
end-to-end history independent characteristics.
The difference between history independent implementations				
for simple ADTs, such as stacks and queues \cite{golovinthesis} versus a complete system,
such as a database, or a file system is a matter
of often exponentially increasing complexity. Fundamentally,
any system can be modeled as an ADT and a history independent implementation
can be sought for the system.

We introduce a general recipe for building history independent systems as follows:

\begin{enumerate}
\item	Model the system as an ADT. For a specific example of file
	system as an ADT, refer to section \ref{hifs:concepts}.
\item	Select a machine model for implementation. While defining the
	machine state identify all machine components that the adversary
	has access to and define the machine state associated with the
	adversary-accessible components.
\item	Depending on the application scenario, fix a desired notion
	of history independence and the corresponding $\Delta$ function.
\item	Based on the definition of $\Delta$, provide an implementation
	over the selected machine model. For complex systems, the implementation
	will likely require the most effort since the machine programs
	implementing the ADT operations must provably ensure that
	the advantage of the adversary is negligible in the $\Delta$HI game.
\end{enumerate}

Using this recipe, %outlined in Section \ref{hitheory:theory2practice:buildinghisystems},
we design, implement, and evaluate a history independent file system (HIFS) and a delete agnostic file system (DAFS).

In Sections \ref{hifs:intro} - \ref{hifs:experiments}, we describe HIFS, an SHI implementation
for file systems.
In Section \ref{hifs:dhi}, we introduce DAFS (delete agnostic file system).
DAFS extends HIFS beyond SHI to implement new history independence notions.
DAFS aims to be more efficient for scenarios in which canonical representations can be avoided.
Further, DAFS extends functionality and resilience of the FS. 
%in which we re-design HIFS to support new history independence notions other than SHI.
%DAFS is a highly efficient history independent implementation
%targeted at specific scenarios for which SHI with canonical representations is 
%an overkill.
%
\subsection{HIFS Overview}
\label{hifs:intro}
Existing file systems, such as Ext3 \cite{ext3} are not history
independent because they organize data on disk
%they produce are 
as a function of both files' data and the sequence of file
operations. The exact same set of files can be organized
differently on disk depending on the
sequence of file system operations that created the set.
%Existing file systems are therefore history dependent.
As a result, observations of data organization on disk
can potentially reveal file system's history.
%, such as evidence of a file delete operation performed in the past.
Moreover, file system meta-data also contains historical information, such as
list of allocated blocks.
Therefore, when observations of data organization
are combined with file system meta-data, and with knowledge of application logic,
significantly more historical information can be derived,
for example, full recovery of deleted data.
It is therefore imperative to hide file system history.

File system history can be hidden by making
%The solution then is to make
file system implementations history independent.
A straight-forward way to achieve this is to use existing
history independent data structures to organize files' data on disk.
%However, several challenges lie on the path towards efficient history
%independent file system implementations. First, individual media
%characteristics render the existence of a single history independent
%design for all storage media unlikely. For instance, the linear-ordered
%storage of disks vs. the wear leveling of SSDs. Hence, specific history
%independent designs are needed for each storage media type.
%Second, current systems heavily benefit from (data and time)
%locality at all layers through heavy caching. Any history independent
%implementation will therefore be required to preserve some degree of data locality for
%comparable efficiency.
Current techniques to make history independent data structures persistent
require the use of history independent hash tables \cite{golovinthesis}. 
The history independent hash tables \cite{blellochhashing} in turn use 
uniform hash functions. The use of uniform hash functions distributes
files' data on storage with no consideration to data locality. \\
Modern filesystems exploit data locality for performance by storing logically related
data at nearby physical locations on the storage device. For e.g., blocks of data belonging
to the same file may be stored physically close to each other to reduce seek time
on traditional storage devices with mechanical parts. This significantly reduces the latency for 
file access. Consequently,
%the as is use of current
existing history independent data structures which do not preserve data locality
%which is critical for file system efficiency on disk storage.
%However, existing history independent data structures
%completely destroy data locality. Hence
cannot be used for practical filesystem design.
%Finally, certain file system features such as caching and journaling 
%by definition require some history to be revealed.
%Such features demand file system implementations that can ensure that
%no additional history is leaked other than what is necessary to
%provide the features.

In HIFS, we overcome the challenge of providing history independence
while preserving data locality. 

\medskip
\noindent
\textbf{Model}
\label{hifs:model}

%For file system history independence we assume the
%following deployment and security model.
%------------------------------------------------------------------------------------
%\subsubsection{Adversary}
%
\noindent
We assume an insider adversary with full access to the
system disk.
%Figure \ref{fig:hifs:hifsneed}).
%By forensic analysis of disk contents,
By analyzing data organization on disk,
the adversary aims to derive file system's history. 
%With respect to history independence the following characteristics
%of the insider adversary are of most importance
%\begin{itemize*}
%\item The adversary is computationally-bounded.
%\item The adversary can make multiple observations of disk contents.
%\item The adversary has full access to any random bits of the machine.
%\end{itemize*}
We assume that the adversary can make multiple observations of disk contents.
Recall from Section \ref{hi:canonical} that thwarting such
an adversary requires SHI with canonical representations.
%Moreover, SHI is secure even against a computationally unbounded adversary
%\cite{golovinthesis}.
Hence, HIFS targets canonical representations
for file storage.
%Adversary actions include (but are not limited to) the following.
%(1) Determine the existence and content of records deleted in the past, thereby
%violating regulatory compliance \cite{ficklebase}.
%
%(2) Determine the order of past file operations to subvert privacy in voting applications
%\cite{blellochhashing}.
%
%(3) Compromise history independence of data structures stored within files
%via file system meta-data and disk layouts.
%------------------------------------------------------------------------------------
%\subsubsection{Storage Medium}
%The underlying storage device is assumed to be a mechanical disk drive,
%not flash storage. Refer to \cite{hifs} for discussion of history
%independence on SSD storage.
%------------------------------------------------------------------------------------
%\noindent
%{\bf Files.~}
%All data structures stored within files are history independent
%\cite{23tree,naorcuckoo,crosbyaggrdict,blellochhashing,golovinbtreaps,golovinbskiplist}.

% \noindent
% {\bf Notations.~}
% The hash of message $M$ is denoted by $h(M)$, while || denotes
% concatenation.

\medskip
\noindent
\textbf{Concepts}
\label{hifs:concepts}

%In Section \ref{hitheory:theory2practice:buildinghisystems} we outlined
%a general recipe for building history independent systems.
\noindent
\begin{comment}
In the following, we use the recipe proposed in section \ref{hi:hifs} to design HIFS.
First, we define a file system ADT (Section \ref{hifs:concepts:fsadt}).
Then, we describe the machine model over which we seek
an history independent implementation for the file system ADT (Section \ref{hifs:concepts:ramdiskmodel}).
In Section \ref{hitheory:canonical}
we have already outlined the need of canonical representations
for SHI. 
Hence, the $\Delta$ function for HIFS is 
the same as that for SHI defined in Section \ref{hi:deltahi}.
Finally, we detail our HIFS implementation (Sections \ref{hifs:architecture} - \ref{hifs:experiments}).
\end{comment}

%---------------------------------------------------------------------------
\subsubsection{File System ADT}
\label{hifs:concepts:fsadt}

A file system organizes data as a set of files. 
We consider a file to consist of some meta-data and a bit string.
That is, a file
%$f = \{ m_{f}, \langle b_{i} \rangle_{i = 1,n} \}$, where $m_{f}$ is the file meta-data
$f = \{ m_{f}, b_{f} \}$, where $m_{f}$ is the file meta-data
and $b_{f} \in \{0,1\}^{*}$.
We define a file system ADT using the file type.
Refer to Section \ref{hitheory:preliminaries:ds:graph} for a discussion on ADTs and types.

\smallskip
A file system is an ADT, that is, a pentuple
$(\mathcal{S}, s_{\phi}, \mathcal{O}, \Gamma, \Psi)$, where
\begin{itemize}[noitemsep,nolistsep]
\item	$\mathcal{S} = 2^{\mathcal{F}}$, is the set of states. Here
	$\mathcal{F}$ is the set of all files.
\item	$s_{\phi} \in \mathcal{S}$ is the initial state.
\item	$\Gamma = \mathbb{N} \cup \{0,1\}^{*} \cup (\mathbb{N} \times \mathbb{N} \times \mathbb{N}) \cup (\mathbb{N} \times \mathbb{N} \times \mathbb{N} \times \{0,1\}^{*})$ is the set of all possible inputs to the filesystem operations. In other words, the set of all possible inputs is composed of : All possible inputs to the filesystem close operation $\cup$ All possible inputs to the filesystem open operation $\cup$ All possible inputs to the filesystem read operation $\cup$ All possible inputs to the filesystem write operation. 
\item	$\Psi = \mathbb{Z} \cup (\{0,1\}^{*} \times \mathbb{Z})$ is the set of all possible outputs from the filesystem operations. In other words, the set of all possible outputs is composed of : All possible outputs of filesystem operations $\cup$ All possible output for filesystem metadata. 
\item	The set of operations $\mathcal{O} = \{$open, read, write, delete, close$\}$, such that
	\begin{itemize*}
		\item	open : $\mathcal{S} \times \{0,1\}^{*} \rightarrow \mathcal{S} \times \mathbb{Z}$.
		\item	read  : $\mathcal{S} \times \mathbb{N} \times \mathbb{N} \times \mathbb{N} \rightarrow \mathcal{S} \times \{0,1\}^{*} \times \mathbb{Z}$.
		\item	write  : $\mathcal{S} \times \mathbb{N} \times \mathbb{N} \times \mathbb{N} \times \{0,1\}^{*} \rightarrow \mathcal{S} \times \mathbb{Z}$.
		\item	delete : $\mathcal{S} \times \mathbb{N} \rightarrow \mathcal{S} \times \mathbb{Z}$.
		\item	close : $\mathcal{S} \times \mathbb{N} \rightarrow \mathcal{S} \times \mathbb{Z}$.
	\end{itemize*}
\end{itemize}

\smallskip
File systems including HIFS support several additional operations.
We have included only a small subset of the operations here for brevity.
%---------------------------------------------------------------------------
\subsubsection{RAMDisk Machine Model}
\label{hifs:concepts:ramdiskmodel}

In Section \ref{hitheory:preliminaries:rammodel} we introduced the RAM
machine model. The RAM model consists of two components, a central processing
unit (CPU), and a random access memory (RAM). However, a file system is generally
used to store and manage data over a secondary storage device. 
Hence, we define the RAMDisk model which in addition to the CPU and memory
also includes the storage disk.

\begin{definition}
\label{def:ramdiskmodel}
{\em RAMDisk Machine Model} \\
A RAMDisk machine model $\mathcal{M_{D}}$ with $m$ $b$-bit memory words, $n$ $b$-bit CPU registers,
and $c$ $k$-bit disk blocks
is a pentuple $(\mathcal{S}, s_{\phi}, \mathcal{P}, \Gamma, \Psi)$, where
$\mathcal{S} = \{0, 1\}^{b(m+n) + c \cdot k}$ is a set of machine states,
$s_{\phi} \in \mathcal{S}$ is the initial state,
$\mathcal{P}$ is the set of all programs of $\mathcal{M_{D}}$,
$\Gamma = \{0, 1\}^{*}$ is a set of inputs,
$\Psi = \{0, 1\}^{*}$ is a set of outputs;
each program
%$p \in \mathcal{P}$ is a function $p : \mathcal{S} \times \Gamma \rightarrow \mathcal{S} \times \Psi$.
$p \in \mathcal{P}$ is a function $p : \mathcal{S} \times \Gamma_{p} \rightarrow \mathcal{S} \times \Psi_{p}$,
where $\Gamma_{p} \subseteq \Gamma$ and $\Psi_{p} \subseteq \Psi$.
\end{definition}

$\mathcal{M_{D}}$ is initialized to state $s_{\phi}$. When a program $p \in \mathcal{P}$ with input $i \in \Gamma_{p}$ is executed by the CPU when $\mathcal{M_{D}}$ is in state $s_{1}$, $\mathcal{M_{D}}$ outputs $\tau \in \Psi_{p}$ and transitions to a state $s_{2}$. This transition is denoted as $p(s_{1},i) \rightarrow (s_{2},\tau)$.
%$\mathcal{M_{D}}$ is initialized to state $s_{\phi}$. When a program $p \in \mathcal{P}$ with input $i \in \Gamma$ is executed by the CPU when $\mathcal{M_{D}}$ is in state $s_{1}$, $\mathcal{M_{D}}$ outputs $\tau \in \Psi$ and transitions to a state $s_{2}$.
%This transition is denoted as $p(s_{1},i) \rightarrow (s_{2},\tau)$.

According to our model (Section \ref{hifs:model}), the adversary has access to the
storage disk.
%Recall from Section \ref{hitheory:theory2practice:components} that
For the purpose of history independence, we need to consider the machine states associated
with the adversary-accessible components only. Hence, from this point onwards we refer to the
storage device state as the machine state.
Since the adversary does not access CPU and RAM components we permit the CPU and RAM
states to reveal history.
%That is, the set of all machine states $\mathcal{S} = \{0, 1\}^{c \cdot k}$,
%where $c$ is the number of $k$-bit disk blocks.
%In our model an adversary attempts to derive file system history via analysis
%of disk contents. Our goal is to prevent any additional
%history to be revealed other than what is evident from
%the file system ADT states.  
%In current work, we are addressing in-memory history independent
%in Section \ref{hifs:discussion} we discuss the case when
%the adversary has access to both the system disk and memory.
%---------------------------------------------------------------------------
\subsubsection{File System Implementation (Data Structure)}
\label{hifs:concepts:fsimpl}

%HIFS is an implementation of the file system ADT defined above.
The objectives of HIFS design are three-fold.
\begin{enumerate*}
\item	For a given set of files, the organization of files' data and files' meta-data
        on disk must be the same independent of the sequence of file
	operations. That is, file system implementation must be canonically
	represented and thereby preserve SHI.
\item	Despite history independent storage, data locality must be preserved.
\item	The implementation must be easily customizable to suit a wide range of data locality scenarios.
\end{enumerate*}

HIFS is a history independent implementation of the file system ADT from Section \ref{hifs:concepts:fsadt}. 
That is, HIFS is a data structure
$\mathcal{D} = (\alpha, \beta, \gamma, s_{0}^{\mathcal{M}})$ 
obtained as follows.
\begin{itemize*}
%\item	$\alpha : \mathbb{N}_{b} \cup \{0,1\}^{b(m+n)} \rightarrow \{0, 1\}^{b(m+n) + c \cdot k}$.
\item	
	For all $n \in \mathbb{N}_{b}$, $\alpha(n) \in \{0,1\}^{b}$.
	Here, $\mathbb{N}_{b} = \{ x | x \in \mathbb{N}\;and\;x \leq 2^{b} \}$,
	$b$ is the machine word length, and
	$\alpha(n)$ is the bit string representing $n$.
	For all $t_{s} \in \{0,1\}^{c \cdot k}, \alpha(t_{s}) = t_{s}$ where $t_s$ represents the current disk state in the RAMDisk model.
	For all $(n_{1}, n_{2}, n_{3}) \in \mathbb{N}_{b} \times \mathbb{N}_{b} \times \mathbb{N}_{b},
	\alpha((n_{1}, n_{2}, n_{3})) = \alpha(n_{1}) || \alpha(n_{2}) || \alpha(n_{3})$.
	For all $(n_{1}, n_{2}, n_{3}, t_{s}) \in \mathbb{N}_{b} \times \mathbb{N}_{b} \times \mathbb{N}_{b} \times \{0,1\}^{c \cdot k}$,
	$\alpha((n_{1}, n_{2}, n_{3}, t_{s})) = \alpha(n_{1}) || \alpha(n_{2}) || \alpha(n_{3}) || t_{s}$.
\item	%$\beta : \mathbb{Z}_{b} \cup \{0,1\}^{*} \rightarrow \{0, 1\}^{b(m+n) + c \cdot k}$.
	%Here, $\mathbb{Z}_{b} = \{ x | x \in \mathbb{Z}\;and\;x < 2^{b} \}$
	%We rely on the underlying operating system representations for $\mathbb{Z}_{b}$.
	For all $z \in \mathbb{Z}_{b}$, $\alpha(z) \in \{0,1\}^{b}$.
	Here, $\mathbb{Z}_{b} = \{ x | x \in \mathbb{Z}\;and\;x \leq 2^{b} \}$,
	$b$ is the machine word length, and
	$\alpha(z)$ is the bit string representing $z$.
\item	$\gamma: \mathcal{O} \rightarrow \mathcal{P}^{\mathcal{M}}$.
	%$\gamma$ is bijective,
	%We provide a machine program $p$ is provided for each file system operation $o$.
	The programs that we provide for each file system operation are the
	key to achieving SHI. We discuss the HIFS programs in Section \ref{hifs:architecture}.
	%We detail the HIFS programs in Section \ref{hifs:architecture}, and
	%evaluate the performance of our implementation in Section \ref{hifs:experiments}.
\item	The initial data structure state $s_{0}^{\mathcal{M}}$ corresponding to the initial
	file system ADT state is obtained by initializing all file system meta-data.
	\footnote{File system meta-data includes superblock, group descriptors,
	inode tables, and disk buckets map. Refer to \cite{hifs} for detailed HIFS
	architecture. The loading of
	file system programs and memory management are done by the operating system.}.
%\item
\end{itemize*}

\subsection{Architecture}
\label{hifs:architecture}

\subsubsection{Overview}
\label{hifs:architecture:overview}
\begin{figure}[t]
\begin{center}
\vspace{1.2cm}
\includegraphics[scale=0.26]{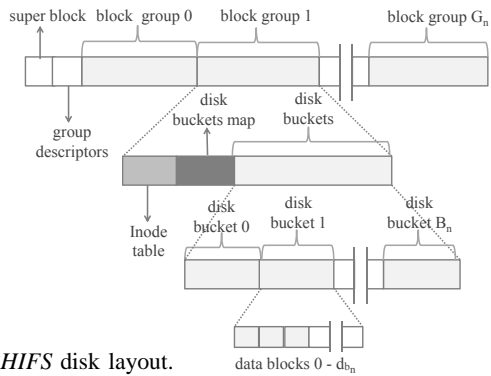}
\vspace{-2.6cm}
\caption[HIFS disk layout]{\begin{small}\emph{HIFS} disk layout. \end{small}
                                                                 
%Key parameters: 
%$\mathsf{G_{n}}\leftarrow$ number of block groups,
%$\mathsf{B_{n}}\leftarrow$ number of disk buckets per block group,
%%$\mathsf{I_{n}}\leftarrow$ number of inodes in inode table,
%$\mathsf{d_{s}}\leftarrow$ data block size in bytes,
%$\mathsf{d_{b_{n}}}\leftarrow$ number of data blocks per disk bucket.
\label{fig:hifs:disklayout}}
\vspace{-0.4cm}
\end{center}
\end{figure}

A file system ADT state contains two pieces of information for each file --
the file meta-data and the file data.
%Each piece of information needs to be represented in the underlying
%machine state (disk layout).
HIFS supports SHI by providing unique memory representations for each ADT state.
%do this by dedicating a portion of the disk storage
%to each. Moreover,
%that both files' meta-data, and files' data are stored on
%disk in canonical form (i.e., uniquely represented).
To ensure unique representations, we first select an existing SHI data structure implementation for
a hash table ADT (Section \ref{hifs:architecture:hiht}).
Then, we re-design the hash table implementation to endow it with
data locality properties (Section \ref{hifs:architecture:keyinsights}).
Finally, we use two instances of re-designed hash table implementation to
store data on disk, one for files' meta-data and the other
for files' data (Section \ref{hifs:architecture:filestorage}).
%In addition, the HIFS implementation maintains file system level meta-data.
%This is separately handled as discussed in Section \ref{hifs:architecture:superblock}.

%\emph{HIFS} closely resembles existing Linux file systems such as Ext2
%\cite{ext2}, exposing the exact same API and utilizing a similar disk
%structure (Figure \ref{fig:hifs:disklayout}).
	%The key differences that give
	%it history independent characteristics are the use of new
	%locality-preserving history independent data structures for all file system
	%meta-data (e.g., the inode table, Section
	%\ref{hifs:architecture:inodetable}) and the fact that the allocation of free
	%disk blocks to files is not based on history. 
%The key difference is that unlike Ext2 the allocation of disk blocks
%to files is not based on history.
%Hence \emph{HIFS} does not use indirect and double indirect blocks to
%map file blocks to disk blocks. 
%Instead, the entire data blocks section on disk is managed as a history
%independent data structure to allocate blocks to files (Section
%\ref{hifs:architecture:filestorage:hi}).

We refer the reader to \cite{hifs} for detailed HIFS architecture.
In the following we focus only on the key features that make HIFS history independent
and locality preserving.

\subsubsection[History Independent Hash Table]{History Independent Hash Table \cite{blellochhashing}}
\label{hifs:architecture:hiht}
%
%The key feature of \emph{HIFS} is the replacement of all
%file system disk structures with history independent versions that we then
%endow with data locality preservation properties.
The SHI data structure of choice is the history independent hash table from
\cite{blellochhashing}.
%Hence, first we describe the hash table
%construction and in subsequent sections illustrate its use in various
%\emph{HIFS} components.
The hash table in \cite{blellochhashing} is based on the stable
matching property of the \emph{Gale-Shapley Stable Marriage} algorithm \cite{galeshapley}.

\smallskip
\noindent
\emph{SHI Hash Table: }
\cite{blellochhashing} uses the stable matching property to construct a canonically represented SHI hash table . 
%(1) The set of keys are considered as the set of men. 
%(2) The set of hash table buckets are considered as the set of women.
%(3) Each key has an ordered preference of buckets and vice versa.
%(4) The preference order of each key is the order in which the buckets are probed
%for insertion, deletion, and search.
%(5) In case of a collision between two keys, the key which ranks higher on the
%bucket's preference list takes the slot. The lower ranked key is relocated to the next
%bucket in its preference list.

The construction in \cite{blellochhashing}
ensures that for a given set of keys, the hash table
data structure state is the same irrespective of the sequence of key insertions
and deletions, thereby making the hash table data structure
canonically represented\footnote{Refer to \cite{blellochhashing} for proof of
canonical representation.}.

\subsubsection{Key Insights}
\label{hifs:architecture:keyinsights}
% 
%The generic insert, delete and search algorithms are listed in Procedure Set \ref{hifs:algorithm:hiht}.
% 
The SHI hash table of \cite{blellochhashing} can be used as is to
organize file's meta-data and files' data on disk. This will
yield a SHI file system implementation.
However, doing so does not preserve data locality, which is an
important goal in HIFS design.
Then a key observation in this context is the following. In the Stable
Marriage algorithm each man in $M$ can rank the $n$ women in $W$ in $n!$
ways and vice-versa. Hence, several sets of preferences from keys to
buckets and buckets to keys are possible, each resulting in a distinct hash
table instance. Therefore, by changing the preference order of keys and
buckets we can control the organization of keys within the hash table.

To enable the re-ordering of preferences
% leads to the realization that we can
we re-write the algorithms of \cite{blellochhashing}.
%to enable easy-custom selection of
%data locality scenarios with minimal modifications. For this,
%we categorize the hash table operations
The re-write categorizes hash table operations in two Procedure Sets, a
\emph{generic} set and a \emph{customizable} set.
The generic procedures implement the overall search, insert, and delete
operations, and can be used unaltered for all scenarios. 
The customizable procedures determine the specific key and bucket
preferences thereby governing data organization, including canonical
representations and data locality.
%the resultant hash table layouts.
%The generic procedures include INSERT, SEARCH and DELETE, listed in
%Procedure Set \ref{hfs:algorithm:hiht}. These in turn use the
%customizable procedures, which include GET\_MOST\_ PREFERRED\_BUCKET,
%GET\_NEXT\_BUCKET, BUCKET\_PR EFERS and KEY\_PREFERS, listed in
%Procedure Set \ref{hifs:algorithm:custexample}.
%
This new procedure classification and rewrite enables HIFS to realize
different data locality scenarios
%(e.g., sequential access, file co-location, etc.)
for the same data set through modifications of
the customizable procedure set only.
The generic procedures and the overall file system operations remain unchanged.
We note that this customization is achieved while preserving SHI.

Due to space constraints we refer the reader to \cite{hifs}
for complete listing of generic and customizable procedures for several data locality scenarios.
In this paper we focus on the scenario
of block group locality. Under block group locality,
it is desired that blocks of the same file are located
close together on disk ideally within the same block group.

\subsubsection{File Storage}
\label{hifs:architecture:filestorage}
%
%\input{write}
%\input{custsequential}
%\input{custnonfilesystem}
%
%\subsubsection{Disk Layout}
%\label{hifs:architecture:filestorage:disklayout}
%\subsubsection{Disk Buckets}
%\medskip
%\noindent
%\emph{Disk Buckets: }
%Refer to Figure \ref{fig:hifs:disklayout} for HIFS disk layout.
File data is stored in blocks on disk. The blocks are grouped into
fixed-size units.
% such that each unit consists of a fixed number of data blocks.
Each unit is termed as a disk bucket (Figure \ref{fig:hifs:disklayout}).
%HIFS relies on a special region in each block group referred to as the 
%disk buckets map for allocation of new disk buckets to files, and for locating
%disk blocks in read and write operations.
Like Ext3 \cite{ext3}, HIFS divides the disk into block groups.
Each block group contains an inode table, a disk buckets map, and a 
set of disk buckets.
Each entry within the disk buckets map has a one-one mapping to the corresponding
disk bucket within the same block group. 
The entry in the disk bucket map contains meta-data about the corresponding disk bucket, such as
whether the bucket is free or occupied.
\subsubsection{Achieving SHI With Data Locality}
\label{hifs:architecture:hi}
%\noindent
%{\bf History Independence.~}
% 
Existing file systems, such as Ext3 \cite{ext3} maintain a list of allocated
blocks within the file inode, which renders the disk space allocation history
dependent. HIFS on the other hand does not rely on allocation lists. Instead,
in HIFS, locations of data blocks are 
%derived directly from file attributes.
%Thus, for each read or write operation in HIFS, the locations of data blocks on disk are
determined using only the current operation parameters and do not
depend on past operations.

In HIFS, the disk bucket maps from all block groups are collectively 
treated as a single SHI hash table.
Then, to achieve canonical representations for file system ADT states,
the file system operations are translated in to SHI hash table
operations as follows:
(a) Keys are derived from the full file path, and from read or write
offset parameters to file system operations.
(b) The hash table buckets are the disk buckets map entries.
(c) Key preferences are set such that each key first prefers all buckets from
one specific block group in a fixed order. Then, buckets from the next adjacent
block group and so on. This ensures that with high probability blocks of the
same file will be located within the same block group.
(d) Buckets prefer keys with higher numerical values. 

The above translation realizes one data locality scenario referred to
as block group locality. In \cite{hifs}
we demonstrate several other scenarios such as sequential file storage and
locality based on external parameters. 

\subsection{Experiments}
\label{hifs:experiments}

\begin{figure}[t]
  %\centering
  %\vspace{-0.8cm}
  \hspace{0.6cm}
  \begin{minipage}[c]{0.38\textwidth}
    %\centering
    %\input{hifs/parameters.tex}
  %\caption{caption here.}\label{fig:sor-v}
  \end{minipage}
  \hspace{1cm}
  \begin{minipage}[c]{0.58\textwidth}
    \includegraphics[scale=0.55]{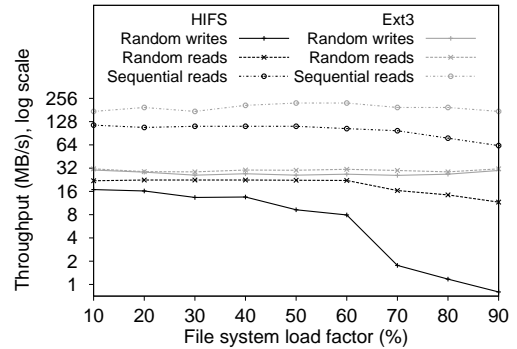}
  \end{minipage}
  
  \caption[]{
  \begin{footnotesize}
  HIFS experimental throughputs.
Load factor = space utilization.
\end{footnotesize}}

\label{fig:hifs:exps}
\vspace{-0.6cm}
\end{figure}
%
%\subsubsection{Setup}
%% \subsection{Setup}
%\label{hifs:exp:setup}
%%\subsubsection{Platform}
%
A detailed evaluation of HIFS for different application profiles and
data locality scenarios is available in \cite{hifs}.
Here, we only list partial results (Figure \ref{fig:hifs:exps}) to give a sense of 
throughputs that can be achieved under SHI.
%
%All experiments were conducted on servers with 8 Intel i7 CPUs at 3.4GHz,
%16GB RAM, and kernel v3.2.0-37. The storage devices of choice
%are Hitachi HDS72302 SCSI drives. The benchmark tool used is Filebench
%\cite{filebench}.
%
%%\subsubsection{Implementation}
%%
%\emph{HIFS} is implemented as a C++ based user-space Fuse \cite{fuse} file system.
%All data structures, including customizable SHI hash tables were
%written from scratch. HIFS implementation is $\approx$10K LOC. 
%
%%\subsection{Measurements}
%%
%Each test run commences with an empty file system, then creates and writes
%new files to storage. The number of files stored is increased until the
%file system is 90\% full. Throughputs are measured at specific load factors
%(disk space utilization) ranging from 10\% to 90\%.
%The writes and subsequent read operations were separated by a complete
%clearing of the system cache. This was done to minimize the effect of
%caching since file system cache is not yet designed to be history
%independent.
%The system cache is cleared before each run.
%
%\subsubsection{Results}
%\label{hifs:exp:results}

%We evaluate HIFS for a database and a web server application profile
%(Table \ref{hifs:table:expparameters}).
We tested HIFS on a file system of size 100 GB with
mean file size 1 GB. The experiments were conducted for different load factors denoted by 
L\footnote{Load factor is the file system disk space utilization.}. The tests were conducted on L$\times$100 files. 
We used 4 KB disk blocks with 8 block groups and 5120 disk blocks per bucket. 
The typical inode size used was 281 bytes and IO size was 32 KB.

The performance of HIFS for read operations is comparable to read
throughputs of Ext3 for load factors up to 60\%. For higher load factors
the write operations sustain significant overheads. This is because
as the load factor increases, the per write overhead to maintain
canonical representations increases exponentially.
%due to hash table collisions.
The overhead is the relocation of existing files' data when a new
file is being written or when a file is being resized.
However, once the write operations achieve canonical representations
with block group locality reads are efficient.

\section{Delete Agnostic File System (DAFS): Journaling and DAHI}
\label{hifs:dhi}

The HIFS implementation (Sections \ref{hifs:architecture} - \ref{hifs:experiments})
supports SHI.
As seen from the experimental results (Figure \ref{fig:hifs:exps}), for higher file
system load factors, write efficiency is low. 
This is because SHI strictly requires canonical representations.
To ensure canonical representations, HIFS relocates data on each write operation.
The amount of data re-located increases exponentially with the file system
load factor. Hence, the write throughputs are significantly lower for load factors
greater than 60\%.

Applications that do not require SHI can be made highly efficient
using new targeted history independence notions. 
In Section \ref{hi:deltahi}, using $\Delta$HI we have defined several
new history independence notions that unlike SHI, do not
require canonical representations.
We have re-designed the file system layer to support such new notions.
Further, we have extended both functionality and resilience of the file system.
The new file system is called \emph{Delete Agnostic File System (DAFS)}.

%In the following we describe how DAFS provides
%two new history independence notions -- revealing last $k$ operations
%for journaling which we refer to as journaled history independence (JHI),
%and delete-agnostic history independence (DAHI) for regulatory compliance.
\medskip

%DAFS is ongoing work. The description here serves only as a preview of things to come. 
%Due to space constraints we give only the key ideas behind the new
%DAFS design.
% and the potential benefits for efficiency.
%----------------------------------------------------------------------------------
\noindent
\textbf{Journaled History Independence (JHI).}
%: Reveal Last $k$ Operations}
\label{hifs:dhi:journaling}
In the event of a system failure, it is imperative that the file system state
is not corrupted. To ensure this, file systems typically employ a journal.
File system operations are first recorded in the journal and then
applied to the file storage area. If a failure occurs while writing
to the journal, the operations can be ignored on system recovery.
On the other hand, if failure occurs while writing to file storage,
then on recovery the operations can be re-played from the journal. 
Thus each write request to the file system causes two disk writes,
one to the journal and one to file storage.
%In current file systems journaling is therefore considered an overhead.

%Recall from Section \ref{hi:deltahi:lastk} that journaling by definition reveals
%the last $k$ operations. This is a necessary tradeoff between resilience and security.
\medskip

%-------------------------------------------------------------------------------
\noindent
\textbf{DAFS Journaling.}
\label{hifs:dhi:journaling:how}
%
%We will consider a logical journal.
In DAFS, a separate region on disk is reserved for a journal in the
form of a circular log.
The journal contains information for a finite number of file system operations,
say $k$ operations.
Operations are recorded in the journal in the order in which they are received by the
file system. To restore consistency after
system failure, it is essential to maintain operation order.  
Hence, the sequence of $k$ operations recorded in the journal cannot be hidden.
The file storage areas, such as the inode tables,
disk bucket maps, and the disk buckets provide SHI just as in the case of HIFS.
Hence, once a file system operation
is applied to file storage and removed from the journal, its timing can no longer
be identified.
\medskip
%In summary, DAFS journaling provides consistent failure recovery whil
% recorded in the journal.

%A journal is often co-located on disk along with files' meta-data and files' data,
%and can record a.
%-------------------------------------------------------------------------------
\noindent
\textbf{Apparent paradox: why journaling increases efficiency}
\label{hifs:dhi:journaling:efficiency}
%Journaling is a significant overhead for existing file systems requiring
%an additional disk write for each file system write request.
%For SHI however, journaling can be used to improve performance.
History independence relaxations that come with journaling allow
significantly more efficient file system operations due to batching.
This is explained in the following.

To maintain canonical representations in HIFS, data is potentially
re-located on each file system write operation. The frequency of data re-location 
increases exponentially with the file system load factor.
Hence, for higher load factors, the number of disk writes for
each write request to the file system is much greater than the two disk writes required for
journaling.
Further, the same data blocks may be re-located several times in consecutive write operations.   
If write operations can be batched, then the number of times a data block is
re-located can be reduced by avoiding redundant moves.
%This is exactly how DAFS uses the journal for efficiency.

In DAFS, we choose to use the journal not only for failure recovery but also
as a buffer to batch write operations.
Write operations are applied to file storage areas only when the journal
is full. During this process, redundant disk writes are eliminated
significantly improving write throughputs.
%How is journaling implemented in HIFS?
%--------------------------------------------------------------------------------
\subsection{Delete-Agnostic HI (DAHI)}
\label{hifs:dhi:dahi}
%

\begin{comment}
\begin{table}[t]
%\vspace{-0.4cm}
%\begin{center}

\caption{Number of disk bucket writes per file system operation.
$L$ is the file system load factor (space utilization), $0 \leq L \leq 1$.\label{hifs:table:dahi}}{%
\begin{tabular}{|c|c|c|}
%
\hline
{\bf File System Operation} & {\bf SHI (HIFS)}	& {\bf DAHI (DAFS)}\\
\hline
%
File write & $O(1/(1-L)^{3})$ & $O(1)$ \\
\hline
File delete & $O(1/(1-L)^{3})$ & $O(1/(1-L)^{3})$ \\
\hline
\end{tabular}}
%\vspace{-0.2cm}
%\caption[Sample paths from ADT and data structure state transition graphs]
%{Sample paths from ADT and data structure state transition graphs of Figure \ref{fig:hitheory:graphs}.
%\label{hifs:table:paths}}
\vspace{-0.6cm}
%\end{center}
\end{table}
%
\end{comment}

Regulations \cite{cfr240} that are specifically concerned with irrecoverable
data erasure and not with other artifacts of history can be met
by systems that support OAHI for the delete operation.
We refer to this notion of history of independence as
delete-agnostic history independence (DAHI).
As discussed in Section \ref{hi:deltahi:oahi}, unlike SHI, OAHI for deletes
can be achieved without canonical representations.
Relaxing the requirement to noncanonical representations presents significant efficiency benefits.
%Hence, the potential increase in efficiency when providing only DAHI is huge.
%Here, we describe how we support DAHI for file deletes in DAFS. 

To make DAFS preserve DAHI only, we first transform the SHI
hash table \cite{blellochhashing} into an DAHI hash table.
Then, we use the DAHI hash table to organize files' data and files' metadata.
%-------------------------------------------------------------------------------------
\subsubsection{DAHI hash table}
The SHI hash table from \cite{blellochhashing} can be transformed
into an DAHI hash table as follows : \\
The hash table insert operation
is modified to not maintain canonical representations. Instead, the
insert operation uses linear probing \cite{Mehta:2004:HDS:1044879} and inserts a key in the
first available bucket. The SHI hash table delete operation\footnote{For complete listing of SHI hash table operations refer to \cite{hifs}.}
alone provides DAHI.
Deletion of a key from the hash table leaves an empty bucket, say bucket $b_{1}$.
The delete operation then finds a key that prefers
bucket $b_{1}$ more than the bucket it is located in according to the 
the gale-shapely stable marriage algorithm, say bucket $b_{2}$.
If such a key is found it is moved from $b_{2}$ to $b_{1}$
making $b_{2}$ empty. The process is then repeated for bucket $b_{2}$, and so on,
until no key is found for relocation.  
The net effect of this process is that a sequence of hash table operations
that contains a delete operation results in the same hash table state as
an insert-only sequence hiding all evidence of the delete. 

%The original HIFS implementation of Section \ref{hifs:architecture} supports SHI.
%Since SHI implies DAHI, the original HIFS implementation supports DAHI. 
%
%Recall from the HIFS architecture discussion in Section \ref{}
%Like HIFS, DAFS uses the canonically represented SHI hash table from \cite{blellochhashing} to organize files' meta-data
%and files' data (Section \ref{hifs:architecture:filestorage}).
%The key difference between HIFS and DAFS in the use of SHI hash table
%is the following. 
%Specifically, canonical representations are achieved
%via the SHI hash table insert and delete operations (listed in \cite{hifs}).
%In DAFS 
%we require DAHI for deletes only. Hence,
%we do not require the hash table insert operation to maintain
%canonical representations. Since the hash table insert operation is
%used by file system write operation, this completely eliminates
%the overhead of maintaining canonical representations on file
%writes.

%---------------------------------------------------------------------------------------------------------

\begin{thm} DAHI hash table preserves delete-agnostic history independence.
\end{thm}

%We refer the reader to our full online paper \cite{FoundationsofHI} for a formal proof.
%\begin{comment}
\begin{proof}

Consider the $\Delta$ history independence game for operation agnostic history independence for deletes played between a challenger $\mathcal{C}$ and an
adversary $\mathcal{A}$. The ADT considered here is the delete-agnostic hash table. $\mathcal{A}$ selects two sequence of operations : $\delta_{0}$
and $\delta_{1}$ such that $\delta_0$ inserts and subsequently deletes an element $x$ from the hash table. $\delta_1$ 
does not contain any delete operations. To ensure indistinguishability between the two sequence, the delete agnostic hash table
must ensure that applying both the sequences of operations to the same initial ADT state should result in the same final ADT state.

Now consider that applying $\delta_{0}$ brings the hash table from an initial state 
to a given state $s$ with element $x$ placed at position $k$ in the hash table. 
Also consider that $\delta_{0}[j] = I(x)$ and $\delta_{0}[m] = D(x)$, 
that is, the $j^{th}$ operations in sequence $\delta_{0}$ inserts $x$ into the hash table and the $m^{th}$ operation deletes $x$ from the hash table.

Consider the elements inserted into the hash table by the operations in $\delta_0$ upto the $m^{th}$ operation.  
We can divide these elements into three sets as follows 

\begin{enumerate}
 \item $A = \{ y \mid \delta_{0}[i] = I(y),i<j\}$.
 \item $B = \{ y \mid \delta_{0}[i] = I(y),i>j,i<m\}$. Further, $\forall y \in B$, $y$ cannot be mapped
 to position $k$ in the hash table using linear probing.
 \item $C = \{ y \mid \delta_{0}[i] = I(y),i>j,i<m\}$. Further, $\forall y \in C$, $y$ can be mapped
 to position $k$ in the hash table using linear probing.
\end{enumerate}

The three sets are constructed in a way such that the elements of the sets are sorted on the order in which the elements 
are inserted into the hash table. To illustrate, consider a set $S \in \{A,B,C\}$ and two elements $a,b \in S$ such that 
$S_i = a$ and $S_j = b$ where $S_k$ is the $k^{th}$ element
of set $S$. Also consider $\delta_{0}[p] = I(a)$ and $\delta_{0}[q] = I(b)$. Then, the sorted property of the sets implies 
that $i < j$ only if $p < q$.

When the delete operation for $x$ is executed in $\delta_0$, the elements of $A$ and $B$ are not affected due to the design
of the hash table. Further by construction, 
once $x$ is deleted, the first element from $C$ is placed at position $k$ and
all other elements already placed in the hash table are remapped (if necessary). If $C = \phi$, then nothing is written to $k$ after the delete. 
Let $C = \{ c_1,c_2,\ldots,\}$ without loss of generality. Also let $D = \{d_1, d_2, \ldots, \}$ be the elements inserted into the hash table after $x$ was inserted.
Once the delete operation is executed in $\delta_0$, $c_1$ will be placed at position $k$ and the elements in $D$ will be remapped to positions in the hash table as if $x$ was never inserted.
Since, $\delta_0$ does not contain any other delete operations, the resulting state after applying the sequence of operations on the delete-agnostic hash table is equivalent to the
resulting state when an insert-only sequence is applied which does not insert and subsequently delete $x$ from the hash table. 
Note that the definition of the $\Delta$ function for operation agnostic history independence for deletes enforces the adversary to select $\delta_1$ to be exactly such a sequence in step 1 of the game. Hence, applying
$\delta_1$ to the delete agnostic hash table will lead to the same modifications to the hash table as
$\delta_0$. This ensures indistinguishability between the two sequences of operations for the adversary and gaurantees that the adversary cannot with the operation agnostic history independence game for deletes on the delete agnostic hash table with more than negligible advantage.

\end{proof}

%\end{comment}

%---------------------------------------------------------------------------
\subsubsection{DAHI in DAFS}

DAFS uses the DAHI hash table for file storage. The DAHI hash table insert
operation is not required to maintain canonical representations.
Since the hash table insert operation is
used by file system write operation,
%this completely 
the overhead of maintaining canonical representations on file writes
is eliminated.

When a file is deleted in DAFS, for each disk bucket allocated to the file,
the same effect is achieved as that for a key deleted from the DAHI
hash table.
%used for each disk bucket allocated to the file.
As a result, no evidence of a delete remains in the file system state
and DAHI is preserved.

Changing the history independence notion from SHI in HIFS to DAHI in DAFS
has significant potential for efficiency. 
%As shown in Table \ref{hifs:table:dahi}, 
The number of writes to
disk buckets needed for DAHI is significantly lower
as compared to the number of writes needed for SHI.
This is because write operations are no longer required to maintain
canonical representations. As a result, when disk buckets are allocated to a file,
other files' data needs no relocation. The relocation of data was precisely
the reason for lower throughputs of HIFS writes.
%------------------------------------------------------------------------------------
%In the full version of the DAFS paper we will
%%theoretically ground these results,
%describe DAFS architecture in detail and provide experimental comparisons between
%DAFS and HIFS.
%
%In DAFS, DAHI is achieved solely via modifications of the customizable
%procedures (Procedure Set \ref{hifs:algorithm:custexample}). The hash table operations from Procedure Set \ref{hifs:algorithm:hiht}
%and the file system operations (listed in \cite{hifs}) remain unchanged.
%This also demonstrates the benefits of our hash table re-design (Section \ref{hifs:architecture:keyinsights})
%to enable easy customization of history independence characteristics
%via only the customizable procedures. 

%To understand when HIFS can reveal the evidence for a delete,
%we need to understand when the hash table leaks evidence of a delete.
%By re-designing the hash table to eliminate all evidence of delete operation
%we can achieve the same in HIFS.

%sequence with deletes looks like an insert-only sequence
%only customizable shows the power of classification and hash table re-design

\subsection{Experiments}
\label{dafs:experiments}

\begin{figure}[t]
  %\centering
  %\vspace{-0.4cm}
  \hspace{0.6cm}
  \begin{minipage}[c]{0.38\textwidth}
    %\centering
    %\input{hifs/dafsparameters.tex}
  %\caption{caption here.}\label{fig:sor-v}
  \end{minipage}
  \hspace{1cm}
  \begin{minipage}[c]{0.58\textwidth}
    \includegraphics[scale=0.55]{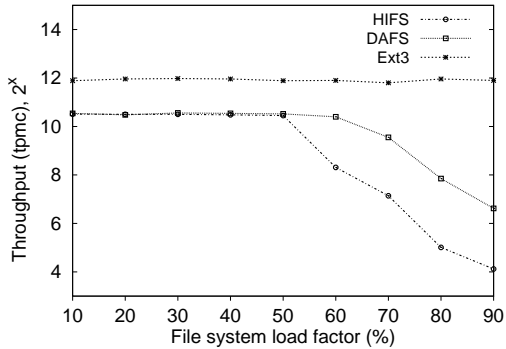}
  \end{minipage}
  
  \caption[]{
  \begin{footnotesize}
  TPCC throughputs for Ext3, HIFS, and DAFS with file system load factor.
  Load factor = space utilization.
  \end{footnotesize}
}
\label{fig:dafs:tpcc}
\vspace{-0.6cm}
\end{figure}

DAFS implements two new history independence notions, JHI and DAHI.
Both JHI and DAHI are aimed to increase file system efficiency.
DAFS can be configured to use DAHI and JHI either exclusively or together.
If DAFS is configured to use both JHI and DAHI, then DAFS uses DAFS
journaling \ref{hifs:dhi:journaling:how}. In this case, 
if the journal contains an entry for a delete operation 
then the adversary can learn about this
delete from the journal. Thus, DAFS allows the user to configure the filesystem
for better performance at the cost of revealing a few deletes to the 
adversary. 
In this section, we compare the performance of DAFS
and HIFS.\\
%-----------------------------------------------------------------------------------
%\textbf{Setup}
%
All experiments were conducted on servers with 2 Intel Xeon Quad-core CPUs at 3.16GHz,
8GB RAM, and kernel v3.13.0-24. The storage devices of choice
are Hitachi HDS72302 SCSI drives.\\
%-------------------------------------------------------------------------------
%\textbf{Implementation}
%
DAFS is implemented as a C++ based user-space Fuse %\cite{fuse} 
file system.
All data structures, including DAHI hash table were
written from scratch. We tested DAFS on a file system of size 10 GB and
mean database size 1 GB. The experiments were conducted on L.10 databases. 
We used 4 KB disk blocks with 4 block groups and 512 disk blocks per bucket. 
The typical inode size used was 281 bytes.\\
%File system setup parameters are listed in Figure \ref{fig:dafs:tpcc}.
%-------------------------------------------------------------------------------
%\textbf{Measurements}
%
To experiment for a real-world scenario we use the TPCC \cite{TPCC} database 
benchmark. The database of choice is Sqlite. Sqlite data files are stored
using HIFS, DAFS (without journaling), and Ext3. The BenchmarkSQL tool is used
to generate the TPCC workload.

Each test run commences with an empty file system and creates
new databases on file system storage. The number of databases is increased until the
file system is 90\% full. The TPCC scale factor is 10 giving a size of 1GB
for each database.
Throughputs are measured at specific load factors.
ranging from 10\% to 90\%.
%To minimize the effect of
%caching, the writes and subsequent read operations are separated by a complete
%clearing of the system cache.
%The system cache is cleared before each run.
%-------------------------------------------------------------------------------
%\textbf{Results}
\label{dafs:exp:results}

Figure \ref{fig:dafs:tpcc} reports the throughputs for HIFS, DAFS, and Ext3.
As per the TPCC specification, throughputs are reported as 
new order transactions executed per minute (tpmc). As seen, the performance
of DAFS is up to 4x times better than HIFS for load factors \textgreater 50\%.
Note that the performance of Ext3 is included as a reference. Ext3 does
not provide DAHI.

For load factors $\leq$ 50\%, HIFS and DAFS exhibit similar performance. 
At lower load factors fewer collisions occur as new files are added to
file system storage. Fewer collisions mean that the frequency of data
relocation to maintain canonical representations is low at load factors
$\leq$ 50\%. Hence, performance of DAFS and HIFS is similar at low
load factors.

\section{Related Work}
\label{hitheory:related}

\begin{table}[t]
\vspace{-0.5cm}
\begin{center}
%\small
\setlength{\tabcolsep}{1pt}.
\caption{Summary of history independent data structures. $\alpha$ $\leftarrow$ load factor,
$N$ $\leftarrow$ number of keys, $B$ $\leftarrow$ block transfer size. Also, 
I : insert, L : lookup, D : delete, R : range
%, P : predecessor, C : compare.
\label{hifs:table:existingresearch}}{%
\begin{tabular}{|c|c|c|c|c|}
\hline
{\bf Data Structure} & {\bf SHI or WHI?} & {\bf Year} & {\bf Ops} & {\bf Runtime}\\
\hline
2-3 Tree \cite{23tree} & WHI & 1997 & I,L,D & $O(\log N)$\\
\hline
Hash Table \cite{naorantipersistence} & SHI & 2001 & I,L & $O(log (1/(1-\alpha)))$\\
\hline
Hash Table \cite{blellochhashing} & SHI & 2007 & I,L,D & $O(1/(1-\alpha)^{3})$\\
\hline
%Ordered Dictionary & SHI & 2007 & I,P,D & $O(\log \log N)$\\
%\cite{blellochhashing} & & & &\\
%\hline
%Order Maintenance & SHI & 2007 & I,C,D & $O(1)$\\
%\cite{blellochhashing} & & & & \\
%\hline
Hash Table \cite{naorcuckoo} & SHI & 2008 & I,L,D & I,D $\rightarrow$ $O(\log N)$, S $\rightarrow$ $O(1)$\\
%& & & & \\
\hline
B-Treaps \cite{golovinbtreaps} & SHI & 2009 & I,D,R & $O(log_B N)$\\
\hline
B-SkipList \cite{bskiplist} & SHI & 2010 & I,D,R & $O(log_B N)$\\
\hline
%R-Trees \cite{rtrees} & WHI & 2012 & I,D,R & n/a\\
%\hline
\end{tabular}}

%\vspace{-0.2cm}
%\normalsize                      
%\caption{Summary of history independent data structures.
%\caption[Summary of history independent data structures]{Summary of history independent data structures. $\alpha$ $\leftarrow$ load factor,
%$N$ $\leftarrow$ number of keys, $B$ $\leftarrow$ block transfer size. Also, 
%I : insert, L : lookup, D : delete, R : range, P : predecessor, C : compare.
%\label{hifs:table:existingresearch}}
\vspace{-0.6cm}
\end{center}
\end{table}
%

%A data structure is history independent if its storage layout is
%a function of the current state and not of the history of past operations
%that led to it.
Existing history independent data structures
%for the RAM model
are summarized in Table \ref{hifs:table:existingresearch}.
%Prior work has designed several history independent
%data structures , summarized in 
%Applications of history independent data structures include
%incremental signature schemes \cite{naorcuckoo}, privacy in voting systems
%\cite{naorcuckoo,blellochhashing,molnarprom,moranwriteonce}, performing
%updates without revealing intermediate states \cite{naorantipersistence},
%debugging parallel computations \cite{blellochhashing}, and reconciliation of
%dynamic sets \cite{naorcuckoo}.
%, and un-traceable deletion \cite{ficklebase}.
%independent data structures to prevent disclosure of erased data.
%summarizes the existing 
%history independent data structures and their run time analysis.
%Other examples include dynamic dependence graphs \cite{acarstaticalgo}
%and Bloom filters \cite{bloomfilter}.
%\normalsize
%
%\subsection{Write Once Storage}
%
The data structures in Table \ref{hifs:table:existingresearch} assume a
re-writable storage medium. \cite{molnarprom} designed a history
independent solution for a write-once medium
%suitable for deployment in voting machines.
The construction is based on the observation from
\cite{naorantipersistence} that a lexicographic ordering of elements in a
list is history independent.  However, write-once memories do not allow
in-place sorting of elements.  Instead \cite{molnarprom} employs copy-over
lists \cite{naorantipersistence}.  When a new element is inserted, a new
list is stored while the previous list is erased.  This requires
$O(n^{2})$ space to store $n$ keys.

%%Another technique suggested in \cite{molnarprom} is to store each new
%\cite{molnarprom} suggests to store each new
%element at a random location on the write-once storage.  In case of
%collisions, a new random location is selected.  Note that although simple
%and space-efficient, this requires the random bits to be hidden from the
%adversary which may not be possible in the targeted scenario involving
%voting machines in poll booths.
\cite{moranwriteonce} improves on \cite{molnarprom} requiring only linear
storage. The key idea is to store all elements in a global hash table
and for each entry of the hash table maintain a separate copy-over list
containing only the colliding elements.

%\subsection{Survey Works}
%% 
%%Various definitions of history independence are analyzed in
%%\cite{hartlinecharacterizinghi}, most relying on canonical representations
%%which are shown to be necessary in achieving it. \cite{heaplowerupper}
%%analyzes the lower and upper bounds on the runtime of the heap and the queue.
%A summary of various data structures from Table \ref{hifs:table:existingresearch} is available in
%\cite{golovinthesis} along with techniques to transform basic data
%structures such as ar, stacks and queues into
%their respective history independent versions.
%
%gohsecureindexes
%23tree, naorantipersistence, naorcuckoo, blellochhashing, golovinbtreaps, golovinbskiplist, rtrees
%
%crosbyaggrdict

%%\input{hifs/discussion}	%do not include
\section{Conclusions}
\label{hifs:conclusion}

In this paper, we took a deep look into history independence
from both a theoretical and a systems perspective.
%we built the theoretical foundations essential
%for designing history independent systems.
We explored the concepts of abstract data types, machine models, data structures
and memory representations. We identified the need for history
independence from the perspective of ADT and data structure
state transition graphs.
Then, we introduced $\Delta$ history independence,
which serves as a general framework to define a broad spectrum
of history independence notions including strong and weak
history independence. We also outlined
a general recipe for building history independent systems
and illustrated its use in designing two history independent
file systems.
%In the following chapter, we will use our theoretical concepts
%and results to design a complete history independent file system.

%\input{hifs/algos}

% Acknowledgments
%\begin{acks}
%The authors would like to thank Dr. Maura Turolla of Telecom
%Italia for providing specifications about the application scenario.
%\end{acks}

% Bibliography

                             % Sample .bib file with references that match those in
                             % the 'Specifications Document (V1.5)' as well containing
                             % 'legacy' bibs and bibs with 'alternate codings'.
                             % Gerry Murray - March 2012

\bibliographystyle{IEEEtran}                           
\bibliography{IEEEabrv,./bib/misc.bib,./bib/hids.bib,./bib/papers.bib,./bib/regulations.bib}        
           
% History dates
%\received{February 2007}{March 2009}{June 2009}

% Electronic Appendix
%\elecappendix

%\medskip

%\section{This is an example of Appendix section head}
%Appendix text here

\end{document}